\documentclass[final,5p,times]{elsarticle}

\usepackage{amssymb}
\usepackage{amsmath}
\usepackage{algorithm}
\usepackage{algorithmic}
\usepackage{hyperref}
\usepackage{graphicx}
\usepackage{multirow}
\usepackage{slashbox,pict2e}
\journal{}

\newtheorem{theorem}{Theorem}
\newtheorem{definition}{Definition}
\newenvironment{proof}[1][Proof]

\begin{document}
\begin{frontmatter}
\title{Branch-and-price-and-cut for the Split-collection Vehicle Routing Problem with Time Windows and Linear Weight-related Cost}

\author[cityu]{Zhixing Luo}
\ead{luozx.hkphd@gmail.com}

\author[hust]{Hu Qin\corref{corl}}
\ead{tigerqin@mail.hust.edu.cn, tigerqin1980@gmail.com}

\author[cityu]{Wenbin Zhu}
\ead{i@zhuwb.com}

\author[cityu]{Andrew Lim}
\ead{lim.andrew@cityu.edu.hk}

\address[cityu]{
Department of Management Sciences,
City University of Hong Kong,\\
Tat Chee Ave., Kowloon Tong, Kowloon, Hong Kong
}

\address[hust]{
School of Management, Huazhong University of Science and Technology,\\
No. 1037, Luoyu Road, Wuhan, China
}

\cortext[corl]{Corresponding author at: School of Management, Huazhong University of Science and Technology, No. 1037, Luoyu Road, Wuhan, China.
Tel.: +852 64117909, +86 13349921096; fax: +86 27 87556437.}

\begin{abstract}
This paper addresses a new vehicle routing problem that simultaneously involves time windows, split collection and linear weight-related cost, which is a generalization of the split delivery vehicle routing problem with time windows (SDVRPTW). This problem consists of determining least-cost vehicle routes to serve a set of customers while respecting the restrictions of vehicle capacity and time windows. The travel cost per unit distance is a linear function of the vehicle weight and the customer demand can be fulfilled by multiple vehicles. To solve this problem, we propose a exact branch-and-price-and-cut algorithm, where the pricing subproblem is a resource-constrained elementary least-cost path problem. We first prove that at least an optimal solution to the pricing subproblem is associated with an extreme collection pattern, and then design a tailored and novel label-setting algorithm to solve it. Computational results show that our proposed algorithm can handle both the SDVRPTW and our problem effectively.
\end{abstract}


\end{frontmatter}

\section{Introduction}
\label{sec:in}
The vehicle routing problem (VRP) and its variants have been extensively studied in literature \citep{Toth02}. These problems consist of designing a set of least-cost routes that fulfill all customer demands and respect a group of operational constraints, such as vehicle capacity, route duration and time windows. The vast majority of the existing vehicle routing models assume that the cost of traversing a route equals the length of that route; a common objective for such models is to minimize the total traveling distance. However, in most practical logistical problems, the real transportation cost depends on many other factors apart from traveling distance, such as vehicle weight, vehicle speed, road conditions and fuel price. Consequently, the distance-minimization vehicle routing models cannot be directly applied by the industrial practitioners that wish to minimize their total transportation costs.

In this study, we only take the effect of vehicle weight and traveling distance on transportation cost into consideration and assume other factors are unchanged. As a result, the transportation cost can be calculated by $d\times f(w)$, where $d$ is the traveling distance and $f(w)$ is a function representing the cost per unit distance paid by the vehicle with weight $w$. Ignoring the effect of vehicle weight is equivalent to setting $f(w)$ to be a constant for any $w >0$. This weight-related cost might have a great impact on the sequencing of the customers along the routes, which is illustrated by the example shown in Figure \ref{fig:1}. This example involves four vertices, i.e., customers $A$, $B$, $C$ and depot $D$, whose locations are vertices of a square with side length 1. It is assumed that $f(w) = 0.08w$, the {\em curb weight} of a vehicle is 5, and the weight demands at customers $A$, $B$ and $C$ are 1, 15 and 1, respectively. Figure \ref{fig:1}(a) indicates a shortest route that incurs a cost of 4.32 while Figure \ref{fig:1}(b) shows a least-cost route with a cost of 4.12. By this example, we can observe that when the weight-related cost is imposed, the customers with more weight demands tend to be served with higher priority.

\begin{figure}[h]
\begin{center}
\resizebox{8cm}{!}{\includegraphics{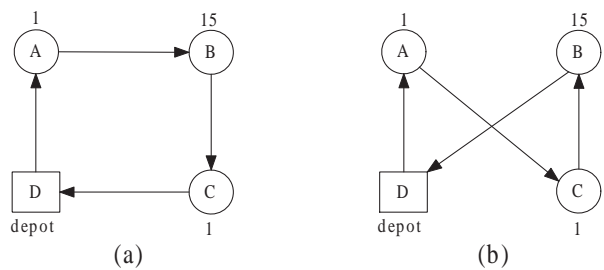}}
\end{center}
\caption{The impact of vehicle weight on customer sequence.} \label{fig:1}
\end{figure}

Numerous real-life applications of the vehicle routing models with weight-related costs arise naturally in Chinese expressway transportation system. As of the end of 2012, over twenty five Chinese provinces have implemented {\em toll-by-weight} schemes in which expressway toll per unit distance is levied according to a monotonically increasing function $f(w)$. Under such toll schemes, a vehicle is charged quite differently when it is empty, normally loaded and overloaded. Moreover, we can also find applications from the transportation service providers who are concerned with fuel consumption and the environmental impact of greenhouse gas (GHG) emissions. The fuel expenditure accounts for a large portion of the overall transportation cost and thus greatly affects the profits of transportation service providers \citep{Xiao2012}. The fuel consumption rate is directly related to vehicle weight; for example, for a vehicle of some type, its fully loaded status might consume more than twice as much diesel fuel as its empty status. In the last decade, the hazardous impacts of GHG, which is directly related to the consumption of fossil fuel, have received growing concerns from the public. Transport sector is one of the key sources of GHG emissions. As revealed by U.S. Greenhouse Gas Inventory Report published in 2012, transportation activities account for 32\% of U.S. CO$_2$ emissions from fossil fuel combustion in 2010. There is a clear tendency that transportation service providers will be forced to undertake the cost of their GHG emissions in the context of new regulations. The cost of fuel consumed or GHG emitted per unit distance by a vehicle with weight $w$ can be represented by a function $f(w)$.

In existing literature, we can only find several prior studies on the vehicle routing models that take vehicle weight into account. The VRP with toll-by-weight scheme was first mentioned by \citet{Shen2009}, who integrated toll-by-weight schemes into the traditional capacitated VRP and developed a simulated annealing algorithm to solve the problem. More recently, \citet{Zhangzz2012} tackled a vehicle routing model that involves toll-by-weight schemes and a single vehicle using a branch-and-bound algorithm. However, this algorithm cannot be adapted to solve the problem with multiple vehicles. \citet{Zhang2011} proposed a multi-depot VRP in which a constant surcharge $C_l$ is incurred for per unit distance per unit weight. Although they did not mention the applications of their problem in the context of Chinese expressway transposition system, imposing this weight-related surcharge is essentially equivalent to levying tolls according to a linear function $f_1(w) = C_l \times w$. The objective of their problem is to minimize the total transportation cost, consisting of distance-related cost, weight-related cost and the fixed cost of dispatching each vehicle. They implemented a scatter search algorithm to solve their problem.

The influence of vehicle weight is perceived in several vehicle routing models that incorporate the costs of fuel and GHG emissions in their objective functions. Most of these models focus on analyzing the influence of vehicle speed, and/or vehicle weight; therefore we can roughly divide them into three classes. The first class of articles only discussed the relationship between vehicle speed and GHG emissions (i.e., ignore the influence of vehicle weight), and investigated methodologies to determine both the route and speed of each vehicle for minimal fuel and emission costs; representative examples include \citet{Palmer2008,Figliozzi2010,Jabali2012}.

The articles in the second class only took vehicle weight into consideration by assuming vehicle speed to be constant. The seminal work of the vehicle routing models that relate vehicle weight to fuel consumption was conducted by \citet{Kara2007}, who introduced an Energy-Minimizing VRP. In this problem, the energy consumed for traversing an edge equals the product of the vehicle weight and the edge length, and the objective is to minimize the total energy rather than the total traveling distance. Lately, \citet{Huang2012} proposed a variant of the VRP with Simultaneous Pickups and Deliveries (VRPSPD) that incorporates the cost of fuel consumption and carbon emissions in the objective function. They assumed without proof that the fuel consumption and carbon emissions per unit distance are both linearly directly proportional to the vehicle weight. Both \citet{Kara2007} and \citet{Huang2012} formulated their problems into mixed integer programming (MIP) models and then solved the models using off-the-shelf MIP solvers. Based on some statistical data, \citet{Xiao2012} derived that the fuel consumption rate can be approximated to a linear function of vehicle weight. They proposed a string-model-based simulated annealing algorithm with a hybrid exchange rule to solve both the distance-minimization VRP and the fuel-minimization VRP. Their experiments on 27 benchmark instances show that the fuel-minimization VRP could help reduce fuel consumption by 5\% on average, compared to the corresponding distance-minimization VRP.

The articles in the last class tackled more general and practical vehicle routing models, where the cost of fuel consumption and GHG emissions is a function of vehicle speed and vehicle weight. The first such model was introduced by \citet{Kuo2010}, who built a fuel-minimization vehicle routing model on the time-dependent VRP (TDVRP) \citep{Ichoua2003,Kuo2009}. In the TDVRP, the time horizon is discretized into a number of intervals. For each edge and each time interval, there is a fixed and known travel speed for all vehicles. The objective of the TDVRP is to minimize the total travel times of all vehicles. In \citet{Kuo2010}, the authors modified the TDVRP by replacing minimizing total travel time with minimizing total fuel consumption. The miles per gallon (MPG) and the gallons per hour (GPH) for an empty vehicle traversing each edge in each time interval are input parameters. They assumed that the fuel consumption rate increases linearly with vehicle weight. For a given routing plan, the total fuel consumed can be easily calculated with the information of MPG, GPH and the vehicle weight on each edge. A simulated annealing algorithm was developed to solve the fuel-minimization TDVRP. \citet{Bektas2011} presented a Pollution-Routing Problem (PRP), which is an extension of the classical VRP with more comprehensive objective function that accounts for the costs of driving, GHG emissions and fuel. The driving cost is linearly directly proportional to the total travel time of all vehicles. The amount of fuel consumed on an edge is approximated as $(\alpha w + \beta v^2)\times d$, where $\alpha$ is an edge-specific constant, $w$ is the vehicle weight, $\beta$ is a vehicle-specific constant, $v$ is the vehicle speed and $d$ is the edge length. It can be easily observed that when the vehicle speed is fixed, the amount of fuel consumed per unit distance is a linear function of the vehicle weight. They formulated the PRP into an integer linear programming model, where the vehicle speed associated with each edge is a decision variable, and then applied CPLEX 12.1 with default settings to solve the model.

The aim of this paper is to address a problem extended from the split-delivery VRP with time windows (SDVRPTW) by modeling the cost per unit distance as a linear function of the vehicle's load weight $w$, i.e., $f(w) = a\times w + b$, where $a$ and $b$ are constant. The SDVRPTW is adapted from the well-studied vehicle routing problem with time windows (VRPTW) by allowing customer demands to exceed the vehicle capacity and relaxing the constraint that each customer must be visited exactly once. We refer the reader to \citet{Ho2004,Desaulniers2010,Archetti2011} for more details of the SDVRPTW. The SDVRPTW can be used to model the cases of delivering goods to or collecting goods from customers. In this article, we consider the collection case and therefore our problem is called the {\em split-collection vehicle routing problem with time windows and linear weight-related cost} (SCVRPTWL). The SDVRPTW is a special case of the SCVRPTWL with $a = 0$ and $b = 1$. Since the combination of several linear functions is also linear, the linear weight-related cost function can be used to model the applications with one or several cost factors, such as traveling distance, linear tolls, fuel consumption and GHG emissions.

The main contributions of this paper are summarized as follows. First, we introduce a more practical and general vehicle routing model that considers the vehicle weight. Second, we provide a branch-and-price-and-cut algorithm for the problem with any type of linear weight-related cost function. In this algorithm, the linear relaxation of the master problem at each branch-and-bound node is solved using a column generation procedure \citep{Desaulniers2005}. Although the master problem of the SCVRPTWL is similar to that of the SDVRPTW presented in \citet{Desaulniers2010}, our pricing subproblem is more complicated compared to that of the SDVRPTW. Thus, we designed a tailored label-setting algorithms to solve the pricing problem of the SCVRPTWL. The dominance procedures used in most of the existing label-setting algorithms are usually based on comparing two labels. However, our label-setting algorithm employs a novel and more efficient dominance procedure that checks whether a label is dominated by a set of labels. Several techniques such as the tabu column generator, decremental search, and bi-directional search are proposed to accelerate the column generation procedure. Third, our comprehensive experimental results show the effectiveness of our proposed algorithm and serve as a baseline for future researchers working on this and other related problems.

\section{Problem Description, Properties and Formulation}
\label{sec:pdf}
The SCVRPTWL is defined on a directed graph $G = (V, E)$, where $V = \{0, 1, \ldots, n, n+1\}$ is the vertex set and $E = \{(i, j)| i, j \in V, ~ i\neq j,~ i\neq n+1,~ j \neq 0\}$ is the edge set. Vertices $0$ and $n+1$ are known as the exit from and the entrance to the {\em depot}, respectively, and the set of remaining vertices $V_C =\{1, \ldots, n\}$ denotes the set of $n$ {\em customers}. Each vertex $i$ is characterized by a positive weight demand $d_i$, a service time $s_i$, and a time window $[e_i, l_i]$ within which the service can be started. For notational convenience, we set $d_0 = 0$, $d_{n+1} = +\infty$, $s_0 = s_{n+1} = 0$, $e_0 = e_{n+1} = 0$ and $l_0 = l_{n+1} = +\infty$. Each edge $(i, j) \in E$ has a nonnegative distance $c_{i,j}$ and a nonnegative traversing time $t_{i, j}$. We assume that both the distances and traversing times satisfy the triangle inequality. We denote by $V^+(i) = \{j \in V| e_i + s_i + t_{i, j} \leq l_j, (i, j)\in E\}$ and $V^-(i) = \{j \in V|e_j + s_j + t_{j, i} \leq l_i, (j, i)\in E\}$ the vertices immediately succeeding and preceding vertex $i$ on graph $G$.

We are given an unlimited number of homogeneous vehicles each with a weight capacity $Q$. Each vehicle is allowed to perform a collection pattern, i.e., it starts from vertex $0$, visits a subset of customers, collects some quantity of products at each visited customer and returns to vertex $n + 1$. A {\em collection pattern} is defined as a route with specified collected quantity at each vertex. If a vehicle arrives at vertex $i$ prior to $e_i$, it must wait until $e_i$ and then starts the service. A collection pattern is feasible if its associated route respects the time windows of all visited customers and its total collected demand does not exceed $Q$. The demand of each customer may be fulfilled by multiple vehicles, i.e., the customer demand may be greater than the vehicle capacity and a customer is allowed to be visited more than once. The traversal cost of edge $(i, j)$ paid by the vehicle with load weight $w_{i, j}$ is calculated by $c_{i, j} \times f(w_{i, j})$, where $f(w_{i, j}) = a \times w_{i,j} + b$ and the intercept $b$ is the cost incurred by the curb weight of the vehicle. The objective of the SCVRPTWL is to find a set of feasible collection patterns such that all customer demands are fulfilled and the total traversal cost is minimized.

We can easily observe a property (called {\em Property} 1) that there must exist an optimal solution in which each route visits each customer at most once. In \citet{Desaulniers2010}, the authors presented a theorem regarding the optimal solutions to the SDVRPTW. We first show that this theorem is also valid for the SDVRPTWL and then derive another two properties. All these three properties can help reduce the search space of the SDVRPTWL significantly.
\begin{theorem}
\label{theorem:1}
Given an instance of the SDVRPTWL where the matrices $[c_{i, j}]$ and $[t_{i,j}]$ satisfy the triangle inequality, there must exist an optimal solution to this instance in which no two vehicles have more than one split customer in common.
\end{theorem}
\begin{proof}
{\em Proof.}
Suppose there exist two collection patterns $p_1$ and $p_2$ that have two common customers $i$ and $j$. The quantities collected at customers $i$ and $j$ in pattern $p_1$ (respectively, $p_2$) are $\delta_i^1$ and $\delta_j^1$ (respectively, $\delta_i^2$ and $\delta_j^2$). Note that $\delta_i^1 + \delta_i^2 \leq d_i$ and $\delta_j^1 + \delta_j^2 \leq d_j$. We increase $\delta_i^1$ and $\delta_j^2$  by $\epsilon$, decrease $\delta_i^2$ and $\delta_j^1$ by $\epsilon$ (see Figure \ref{fig:2}), and do not change the quantities collected at the remaining customers. Obviously, only the costs of edges between customers $i$ and $j$ may be affected by this quantity adjustment. Let $c^1_{i \rightarrow j}$ and $c^2_{i \rightarrow j}$ be the traveling distances from customer $i$ to customer $j$ in patterns $p_1$ and $p_2$, respectively (note that $i\rightarrow j$ may cover more than two vertices). After the adjustment, the cost of $p_1$ will increase by $a \times c^1_{i \rightarrow j} \times \epsilon$ while the cost of $p_2$ will decrease by $a \times c^2_{i \rightarrow j} \times \epsilon$. If $a \times \epsilon \times (c^1_{i \rightarrow j} - c^2_{i \rightarrow j}) \leq 0$, we have a motivation to increase $\delta_i^1$ until either $\delta_j^1 = 0$ or $\delta_i^2 = 0$ for less total cost. If $\delta_j^1 = 0$ or $\delta_i^2 = 0$, we can safely remove customer $j$ from $p_1$ or customer $i$ from $p_2$ without increasing the costs of pattern $p_1$ or $p_2$. We can analyze the case when $c^1_{i\rightarrow j} \geq c^2_{i\rightarrow j}$ in the same manner. Hence, if $p_1$ and $p_2$ exist in an optimal solution, they can be adjusted to have one customer in common without increasing the total cost. $\Box$
\end{proof}

\begin{figure}[!h]
\begin{center}
\resizebox{8cm}{!}{\includegraphics{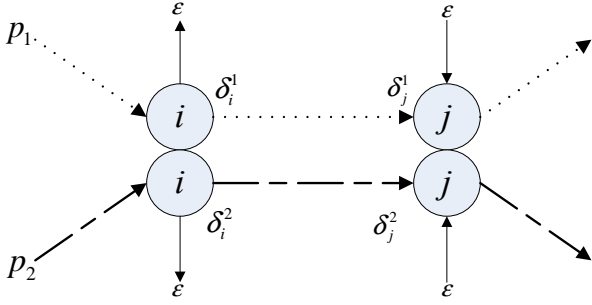}}
\end{center}
\caption{ An example of two patterns that have two customers in common.} \label{fig:2}
\end{figure}

The properties derived from Theorem \ref{theorem:1} are: there must exist an optimal solution to the SCVRPTWL in which
\begin{itemize}
\item at most one vehicle assigned to a route with two or more customers ({\em Property} 2).
\item each edge $(i, j) \in E_C$ appears at most once, where $E_C = \{(i, j)| (i, j)\in E,~ i, j \in V_C\}$ ({\em Property} 3).
\end{itemize}

We now present an arc-flow formulation for the SCVRPTWL, which will be exploited in the Dantzig-Wolfe decomposition proposed in the next section. This formulation uses the following additional notations:

{\bf Parameters}
\begin{itemize}
\item $K = \{1, 2, \ldots, m\}$: the set of $m$ available vehicles.
\item $M$: a sufficiently large positive number.
\end{itemize}

{\bf Decision Variables}
\begin{itemize}
\item $x_{i, j, k}$: the binary variable that equals 1 if vehicle $k$ traverses edge $(i, j)$, and 0 otherwise.
\item $w_{i, j, k}$: the load weight of vehicle $k$ who traverses edge $(i, j)$.
\item $z_{i, j, k}$: the cost paid by vehicle $k$ for traversing edge $(i, j)$.
\item $a_{i, k}$: the service starting time of vehicle $k$ at customer $i$.
\end{itemize}

Using these notations, the SDVRPTWL can be modeled as:
\begin{small}
\begin{align}
\min~& \sum_{k \in K}\sum_{i \in  V_C \cup \{0\}} \sum_{j \in V^+(i)}z_{i, j, k}  \label{eqn:mipaf:obj} \\
\mbox{s.t.}~~& z_{i, j, k} \geq c_{i, j}(a  w_{i,j,k} + b x_{i, j, k}), \forall~ k \in K, ~i \in  V_C \cup \{0\}, ~ j \in V^+(i) \label{eqn:mipaf:c1}\\
&\sum_{k \in K} \Big{(}\sum_{j \in V^+(i)}w_{i, j, k} - \sum_{j \in V^-(i)}w_{j, i, k}\Big{)} \geq d_i, ~\forall~ i \in V_C \label{eqn:mipaf:c3}\\
&\sum_{k \in K}\sum_{j \in V^+(i)}x_{i, j, k} \geq \Bigg{\lceil} \frac{d_i}{Q} \Bigg{\rceil}, ~\forall~i \in V_C\label{eqn:mipaf:c3.1}\\
&\sum_{i \in V_C \cup \{ n + 1\}}x_{0, i, k} = 1, ~\forall~ k \in K \label{eqn:mipaf:c4}\\
&\sum_{j\in V^+(i)}x_{i,j,k} = \sum_{j \in V^-(i) }x_{j, i, k}  \leq 1, ~\forall~ k \in K, ~ i \in V_C\label{eqn:mipaf:c5}\\
&\sum_{i \in V_C \cup \{0\}}x_{i, n+1, k} = 1, ~\forall~ k \in K \label{eqn:mipaf:c6}\\
&w_{i, j, k} \leq Qx_{i, j, k},  ~\forall~ k\in K, ~i \in  V_C \cup \{0\}, ~ j \in V^+(i) \label{eqn:mipaf:c8}\\
&a_{j,k} \geq a_{i,k} + s_i + t_{i, j} + M(x_{i, j, k} - 1), \nonumber \\
&~~~~~~~~~~~~~~~~~~~~~~~~~~~~\forall~ k \in K, ~i \in V_C \cup \{0\},~ j \in V^+(i)\label{eqn:mipaf:c9}\\
& e_i \leq a_{i, k} \leq l_i, ~\forall~ k \in K,~ i \in V\label{eqn:mipaf:c10}\\
&x_{i, j, k} \in \{0, 1\}, ~\forall~ k\in K, ~i \in  V_C \cup \{0\}, ~j \in V^+(i) \nonumber\\
&a_{i, k} \geq 0, ~ \forall ~ k \in K, ~ i \in V \nonumber\\
&z_{i, j, k} \geq 0,~ w_{i, j, k} \geq 0, \forall~ k\in K, ~i \in  V_C \cup \{0\}, ~ j \in V^+(i)\nonumber
\end{align}
\end{small}

The objective function (\ref{eqn:mipaf:obj}) aims at minimizing the total travel distance. The traversal cost of edge $(i, j)$ incurred by the vehicle with load weight $w_{i, j}$ is calculated by Constraints (\ref{eqn:mipaf:c1}). Constraints (\ref{eqn:mipaf:c3}) ensure that the demand of each customer is fulfilled. A minimum number of vehicles that serve customer $i$ is imposed by Constraints (\ref{eqn:mipaf:c3.1}), which are redundant constrains that are used to strengthen the linear relaxation of the model. Constraints (\ref{eqn:mipaf:c4}) -- (\ref{eqn:mipaf:c6}) define the structure of each possible route. The load weight of each vehicle on each edge cannot exceed $Q$ and thus Constraints (\ref{eqn:mipaf:c8}) apply. Constraints (\ref{eqn:mipaf:c9}) and (\ref{eqn:mipaf:c10}) ensure that all customer time windows are respected.

\section{Dantzig-Wolfe Decomposition}
\label{sec:dwd}
We can directly apply CPLEX to handle the arc-flow formulation. Nevertheless, after some preliminary experiments, we find that the size of the instances optimally solved by CPLEX is quite limited. To achieve optimal solutions for instances of practical size, we reformulate the SCVRPTWL into a {\em master problem} through Dantzig-Wolfe decomposition \citep{Dantzig1960} and then develop a branch-and-price-and-cut algorithm to solve it.

Applying Dantzig-Wolfe decomposition to the arc-flow formulation yields a master problem and a {\em pricing subproblem}. The master problem contains a decision variable for each collection pattern. We solve the master problem by a branch-and-bound procedure, where at each branch-and-bound node a lower bound is obtained by column generation procedure and the introduction of violated valid inequalities. The pricing subproblem is solved using a tailored label-setting algorithm.


\subsection{Master Problem}
To present the master problem, we define the following additional notations:

{\bf Parameters}
\begin{itemize}
\item $R^s$: the set of all routes visiting a single customer and satisfying the time window constraint.
\item $R^m$: the set of all routes visiting more than one customer and satisfying all time window constraints.
\item $P_r$: the set of all collection patterns compatible with route $r$.
\item $c_{r, p}$: the cost of collection pattern $p$, where $p \in P_r$.
\item $\alpha_{i,r}$: the binary parameter that equals 1 if customer $i$ is used in route $r$, and 0 otherwise.
\item $\beta_{i, j, r}$: the binary parameter that equals 1 if edge $(i, j)$ is used in route $r$, and 0 otherwise.
\item $\delta_{i, p}$: the quantity collected at customer $i$ in pattern $p$.
\end{itemize}

{\bf Decision variables}
\begin{itemize}
\item $\theta_{r, p}$, the nonnegative integer variable indicating the number of the vehicles assigned to pattern $p$ compatible with route $r$.
\item $\theta_r$, the nonnegative integer (respectively, binary) variable indicating the number of the vehicles assigned to route $r \in R^s$ (respectively, $R^m$). The binary requirement is derived from {\em Property} 2.
\end{itemize}

With the above notations, the master problem (MP) is given as:
\begin{small}
\begin{align}
z^{MP} = &\min \sum_{r\in R} \sum_{p \in P_r}c_{r, p}\theta_{r, p} \label{eqn:mipmp:obj} \\
\mbox{s.t.}~~& \sum_{r\in R}\sum_{p\in P_r}\delta_{i, p}\theta_{r, p} \geq d_i, ~\forall~  i \in V_C \label{eqn:mipmp:c1}\\
&\sum_{r \in R}\sum_{p \in P_r} \alpha_{i, r}\theta_{r, p} \geq \Bigg{\lceil} \frac{d_i}{Q} \Bigg{\rceil}, ~\forall~ i \in V_C \label{eqn:mipmp:c2}\\
& \theta_{r, p} ~\geq 0, ~\forall~ r \in R, ~ p\in P_r \label{eqn:mipmp:c4}\\
&\theta_r = \sum_{p \in P_r}\theta_{r, p}, ~\forall~ r \in R \label{eqn:mipmp:c5}\\
&\theta_r \in \{0, 1\}, ~\forall~ r \in R^m \label{eqn:mipmp:c6}\\
&\theta_r ~~~\textrm{integer}, ~\forall~ r \in R^s\label{eqn:mipmp:c7}
\end{align}
\end{small}
The objective function (\ref{eqn:mipmp:obj}) aims at minimizing the total travel distance. Constraints (\ref{eqn:mipmp:c1}) -- (\ref{eqn:mipmp:c2}) are equivalent to Constraints (\ref{eqn:mipaf:c3}) -- (\ref{eqn:mipaf:c3.1}), respectively. Constraints (\ref{eqn:mipmp:c2}) are redundant constraints that are used to strengthen the linear relaxation of the MP (called {\em LMP} for short). Constraints (\ref{eqn:mipmp:c4}) -- (\ref{eqn:mipmp:c7}) are binary or integrality requirements on the decision variables $\theta_{r, p}$ and $\theta_r$. In this formulation, each variable $\theta_{r, p}$ corresponds to a column composed of parameters $c_{r, p}, \delta_{i, p}$ and $\alpha_{i, r}$.

In practice, even for a small-size instance, the master problem contains a huge number of variables (or columns). Hence, this model cannot be directly handled by CPLEX. With a subset of variables $\theta_{r, p}$, the optimal solution of the LMP can be obtained with the help of the column generation procedure. Therefore, we do not need to enumerate all variables $\theta_{r, p}$ explicitly. Variables $\theta_{r}$ are not required in the LMP but will be used to check whether the optimal solution to the LMP is also optimal to the MP.

\subsection{Pricing Subproblem}
Given a dual solution to the LMP, the pricing subproblem is used to find a master variable $\theta_{r, p}$ (i.e., a collection pattern $p$ compatible with route $r$) that has the least reduced cost. Solving the pricing subproblem is essentially equivalent to enumerating all feasible collection patterns. Below we will use $\pi = (\pi_1, \ldots, \pi_n)$ and $\mu = (\mu_1, \ldots, \mu_n)$ as the values of the dual variables associated with Constraints (\ref{eqn:mipmp:c1}) -- (\ref{eqn:mipmp:c2}), respectively, and set $\pi_0 = \pi_{n+1} = 0$ and $\mu_0 = \mu_{n+1} = 0$ without loss of generality. The reduced cost of a pattern $p$ compatible with route $r$ can be calculated by:
\begin{align}
\bar{c}_{r, p} = c_{r, p} - \sum_{i \in r}(\delta_{i, p} \pi_i  + \mu_i)
\end{align}

Since all vehicles are identical, the pricing subproblem associated with each vehicle can be written as:
\begin{small}
\begin{align}
z^{PS} = & \min\sum_{(i, j) \in E}c_{i, j}(aw_{i, j} + bx_{i, j}) - \sum_{i \in V_C}\pi_i\delta_{i}
- \sum_{i\in V_C}\mu_i\sum_{j \in V^+(i)}x_{i, j} \label{eqn:pp:obj} \\
\mbox{s.t.}~~& \sum_{j \in V^+(0)}x_{0, j} = 1\label{eqn:pp:c1}\\
&\sum_{i \in V^-(n+1)}x_{i, n+1} = 1 \label{eqn:pp:c2}\\
&\sum_{j \in V^+(i)}x_{i,j} = \sum_{j\in V^-(i)}x_{j, i} \leq 1, ~\forall~i\in V_C  \label{eqn:pp:c3}\\
&\delta_i = \sum_{j \in V^+(i)}w_{i, j} - \sum_{j\in V^-(i) }w_{j, i},  ~\forall ~ i \in V_C \label{eqn:pp:c4.5}\\
& \delta_i \leq \min\{d_i, Q\} \sum_{j\in V^+(i)}x_{i, j}, ~\forall ~  i \in V_C \label{eqn:pp:c4}\\
&w_{i, j} \leq Qx_{i, j}, ~\forall ~ i \in V_C \cup  \{0\}, ~  j\in V^+(i) \label{eqn:pp:c5}\\
&a_j \geq a_i + s_i + t_{i, j} + M(x_{i, j} - 1),\nonumber\\
& ~~~~~~~~~~~~~~~~~~~~~~~~~~~~~~~~~~~~~\forall~i \in   V_C \cup \{0\}, ~  j\in V^+(i) \label{eqn:pp:c6}\\
&e_i \leq a_i \leq l_i, ~\forall~ i \in V_C  \label{eqn:pp:c7}\\
&x_{i, j} \in \{0, 1\},~  w_{i, j} \geq 0~,~ \forall~i \in V_C \cup \{0\}, ~  j \in V^+(i)\nonumber\\
&a_i \geq 0, ~\delta_i\geq 0, ~\forall~ i \in V_C \nonumber
\end{align}
\end{small}
where
\begin{itemize}
\item $x_{i, j}$: the binary variable that equals 1 if edge $(i, j)$ is used in the pattern, and 0 otherwise.
\item $w_{i, j}$: the load weight on edge $(i, j)$.
\item $a_i$: the service starting time at customer $i$.
\item $\delta_i$: the quantity collected at customer $i$.
\end{itemize}
The objective function (\ref{eqn:pp:obj}) aims to achieve the minimal reduced cost of all feasible collection patterns. Constraints (\ref{eqn:pp:c1}) and (\ref{eqn:pp:c2}) require that the route must start from vertex $0$ and end at vertex $n+1$. Constraints (\ref{eqn:pp:c3}) ensure that each customer can be visited at most once. Constraints (\ref{eqn:pp:c4.5}) and (\ref{eqn:pp:c4}) state that the quantity collected at customer $i$ cannot exceed $d_i$. Constraints (\ref{eqn:pp:c5}) guarantee that the flow on each edge cannot exceed the vehicle capacity and equals zero if that edge is not used. Constraints (\ref{eqn:pp:c6}) define the relationship between the service starting times of two consecutively visited customers. Constraints $(\ref{eqn:pp:c7})$ ensure that the time windows of all visited customers must be respected.

In many existing branch-and-price algorithms that were developed to solve the VRPTW or other vehicle routing problems, their pricing subproblems are usually {\em elementary shortest path problem with resource constraints (ESPPRC)} \citep{Feillet2004}. Examples can be found in \citet{Desrochers1992,Gutierrez-Jarpa2010,Azi2010} and \citet{Bettinelli2011}. Evidently, it is not appropriate to view our pricing subproblem as an ESPPRC due to the existence of variables $w_{i, j}$. We call our pricing subproblem {\em elementary least-cost path problem with resource constraints} (ELPPRC). Similar pricing subproblems can be found in \citet{Ioachim1998} and \citet{Mattos2012}. The ELPPRC is obviously $\mathcal{NP}$-complete since it can reduce to an ESPPRC by setting $a=0$ and $b =1$ in cost function $f(w)$. This implies that optimally solving the pricing subproblem is computationally expensive. In the next section, we design an ad hoc label-setting algorithm to optimally solve the pricing subproblem.

\section{Column Generation}
Column generation is applied to solve the LMP (i.e., the linear relaxation of the formulation (\ref{eqn:mipmp:obj}) -- (\ref{eqn:mipmp:c7})) augmented by appropriate branching decisions and some cutting planes. For an overview of column generation, the reader is referred to \cite{Desaulniers2005,Lubbecke2005}. The optimal solution value of the LMP is a lower bound of its associated branch-and-bound node. The column generation procedure cannot directly solve the LMP due to its inability of enumerating all variables $\theta_{r, p}$. Instead, it is an iterative procedure that alternates between solving a {\em restricted linear relaxation of the master problem} (RLMP) and a pricing subproblem. The RLMP is the LMP restricted to a subset of all variables $\theta_{r, p}$, which can be optimally solved by the simplex algorithm. The goal of solving the pricing subproblem is to identify the columns that have negative reduced costs with respect to the dual optimal solution of the current RLMP. If no such column is found, the column generation procedure is terminated with an optimal solution to the current RLMP, which is also an optimal solution to the LMP. Otherwise, we introduce one or more columns with negative costs into the current RLMP and restart the column generation iteration.

In this section, we first prove that the optimal solution of the pricing subproblem must be an {\em extreme collection pattern}. Based on this finding, we then develop a label-setting algorithm to solve the pricing subproblem. Finally, several strategies are introduced to accelerate the label-setting algorithm.

\subsection{Extreme Collection Pattern}
We first give the definition of the extreme collection pattern as follows.
\begin{definition}
A collection pattern $p$ is an extreme collection pattern if and only if it is composed of zero collections ($\delta_i = 0$), full collections ($\delta_i = d_i$) and at most one split collection ($0 < \delta_i < d_i$).
\end{definition}
Then, we can prove the following theorems.
\begin{theorem}
Given a route $r$, any collection pattern $p \in P_r$ can be represented by a convex combination of extreme collection patterns in $P_r$.
\end{theorem}
\begin{proof}
{\em Proof.} We assume $r = (v(1), v(2),  \ldots, v(|r|))$, where $v(i)$ is the index of the $i$-th vertex, $v(1) = 0$ and $|r| \geq 2$ is the number of vertices in route $r$. Let $p = (\delta_{v(1)}, \delta_{v(2)}, \ldots,  \delta_{v(|r|)})$ be an arbitrary feasible collection pattern compatible with route $r$. Then, $P_r$ is the feasible region defined by the following $|r|+1$ constraints:
\begin{align}
&\sum_{i = 1}^{|r|}\delta_{v(i)} \leq Q \label{cv:1}\\
&0 \leq \delta_{v(i)} \leq d_{v(i)}, ~\forall~ 1 \leq i \leq |r| \label{cv:2}
\end{align}

It is easy to observe that $P_r$ is a closed convex set. Thus, any point in $P_r$ can be represented by a convex combination of the extreme points of $P_r$, each corresponding to an extreme collection pattern.
There are $|r|$ decision variables $\delta_{i}$ $(1\leq i \leq |r|)$, so we must use $|r|$ active independent constraints to define each extreme point of $P_r$. In other words, only one of $|r|+1$ independent constraints can be loose in an extreme point of $P_r$.

As a result, if one extreme point has a loose constraint, e.g., $0 < \delta_{v(k)} < d_{v(k)}$, it must have $\sum_{i =1}^{|r|}\delta_{v(i)} = Q$, and $\delta_{v(i)}$ equals either $d_{v(i)}$ or zero for all $ v(i) \in r$ except $i = k$. If  $\sum_{i =1}^{|r|}\delta_{v(i)} < Q$, the corresponding extreme point must have either $\delta_{v(i)} =d_{v(i)}$ or $\delta_{v(i)} = 0$ for all $v(i) \in r$.
$\Box$
\end{proof}

\begin{theorem}
One of the optimal solutions to the pricing subproblem must be an extreme collection pattern.
\label{theorem:3}
\end{theorem}
\begin{proof}
{\em Proof.} Assume the optimal solution of the pricing subproblem is a collection pattern compatible with route $r$. If route $r$, namely all variables $x_{i, j}$, is fixed, the pricing subproblem can be written as a linear relaxation of a bounded knapsack problem:
 \begin{align}
 C(r, Q) = \min~ & \sum_{i =1}^{|r|-1} c_{v(i), v(i + 1)}\Big{(}a\sum_{j = 1}^i\delta_{v(j)} + b\Big{)}\nonumber\\
  &- \sum_{i = 1}^{|r|}\bigg{(}\delta_{v(i)}\pi_{v(i)} + \mu_{v(i)}\bigg{)} \label{exp:8}\\
 \mbox{s.t.}~~ & \textrm{Constraints (\ref{cv:1}) and (\ref{cv:2})}. \nonumber
 \end{align}
From this model, we can easily find that one of the optimal solutions to the pricing subproblem must be an extreme point of $P_r$, which represents an extreme collection pattern.
 $\Box$
\end{proof}

According to Theorem \ref{theorem:3}, we can solve the pricing subproblem to optimality by only examining all extreme collection patterns, which significantly reduces the search space of the label-setting algorithm. The objective function (\ref{exp:8}) can be rewritten as:
\begin{align}
& f_r - \sum_{i =1}^{|r|}\delta_{v(i)}g_{v(i)} \label{exp:reduced}
\end{align}
where
\begin{align}
f_r = &\sum_{i = 1}^{|r|-1}bc_{v(i), v(i+1)}  - \sum_{i = 1}^{|r|}\mu_{v(i)}, \nonumber \\
g_{v(i)} = & \pi_{v(i)} - a\sum_{j = i}^{|r|-1}c_{v(j), v(j + 1)}, ~\forall ~ 1 \leq i \leq |r| -1, \nonumber\\
g_{v(|r|)} = & \pi_{v(|r|)}.  \nonumber
\end{align}
This shows that the reduced cost of a collection pattern consists of two components: the first component $f_r$ (called {\em fixed cost}) is a constant only related to route $r$, while the second component $- \sum_{i =1}^{|r|}\delta_{v(i)}g_{v(i)}$ (called {\em variable cost}) is determined by both route $r$ and quantity $\delta_{v(i)}$. The value of $g_{v(i)}$ can be viewed as the profit per unit product collected from vertex $v(i)$.

Given a (partial) route $r$ and the total collected quantity $\hat{q}$ ($0 \leq \hat{q} \leq \sum_{i\in r}d_i$) along this route, the minimal reduced cost $G(r, \hat{q}$) of all possible collection patterns can be computed using a greedy procedure shown in Algorithm \ref{alg:greedy}. The collection pattern associated with $G(r, \hat{q})$ produced by the greedy procedure is obviously an extreme collection pattern. Note that when $\hat{q} > \sum_{i\in r}d_i$, there does not exist feasible collection patterns. When route $r$ is fixed, we can view $G(r, q)$ as a function of $q$, called the {\em reduced cost function}. From Algorithm \ref{alg:greedy}, we observe that $G(r, q$) is a convex and continuous piece-wise linear function of $q$. We illustrate this function in Figure \ref{fig:function}, where $G(r, 0) = f_r$,  the slope $sl_k = g_{v'(k)}$ and $q_k = \sum_{i = 1}^kd_{v'(k)}$. We can directly get the value of $C(r, Q)$ from $G(r, q)$ by: $C(r, Q) = G(r, q^*) = \min_{0\leq q\leq Q}\{G(r, q)\}$, which implies that $\delta_{v(i)}$ is set to zero if $g_{v(i)} \leq 0$. Actually, to compute $C(r, Q)$, we can perform a modified Algorithm \ref{alg:greedy} in which $\hat{q} = Q$ is defined as the allowable capacity and only the vertices with $g_{v(i)} > 0$ are considered.

\begin{algorithm}[!htp]
\caption{The greedy procedure of computing $G(r, \hat{q})$.}
\label{alg:greedy}
\begin{algorithmic}[1]
\STATE INPUT: $g_{v(i)}$ and the total collected quantity $\hat{q}$;
\STATE Set $\delta_{v(i)} = 0$ for all $1\leq i \leq |r|$;
\STATE Sort all vertices in route $r$ by decreasing value of $g_{v(i)}$, yielding a sorted vertex list $(v'(1), v'(2), \ldots, v'(|r|))$.
\STATE The remaining capacity $rc$  $\leftarrow$ $\hat{q}$;
\STATE $ k $ $\leftarrow$ $1$;
\WHILE{$rc \geq 0$ and $k \leq |r|$} \label{alg1:linex}
\STATE Set $\delta_{v'(k)}$ $\leftarrow$ $\min \{d_{v'(k)}, rc\}$, $rc$ $\leftarrow$ $rc - \delta_{v'(k)}$ and $k$ $\leftarrow$ $k+1$;
\ENDWHILE
\STATE Compute $G(r, \hat{q})$ according to Expression (\ref{exp:reduced}).
\end{algorithmic}
\end{algorithm}


\begin{figure}[!th]
\begin{center}
\resizebox{8cm}{!}{\includegraphics{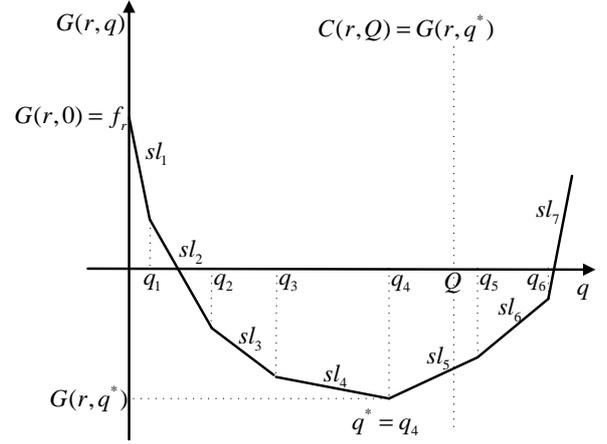}}
\end{center}
\caption{Graphic representation of the reduced cost function $G(r, q)$.} \label{fig:function}
\end{figure}

\subsection{The Label-Setting Algorithm}
\label{sub:lsa}
The label-setting algorithm is a widely used technique to solve the pricing subproblems of various vehicle routing models, such the ESPPRC \citep{Feillet2004,Righini2008}, the shortest path problem with resource constraints (SPPRC) \citep{Irnich2005} and the shortest path problem with time windows and linear node costs \citep{Ioachim1998}. The aim of solving the pricing subproblem is to identify the complete routes with negative reduced cost, namely $C(r, Q)$. In our label-setting algorithm, a multi-dimensional label $E_i = ({\tau}_i, N_i, V_i, G(r, q) = (F_i, {SL}_i, I_i))$ is defined to represent a state associated with a feasible (partial) route $r$ from vertex 0 to vertex $i$, where:
\begin{itemize}
\item $\tau_i$ is the earliest service starting time at vertex $i$, which must lie within $[e_i, l_i]$;
\item $N_i \subseteq V$ is the set of all visited vertices;
\item $V_i \subseteq V $ is the set of all vertices that could be reached from route $r$;
\item $G(r, q)$ is the reduced cost function associated with route $r$, which could be constructed using: $F_i = f_r$, ${SL}_i = \{g_j\}_{ j\in r}$ and $I_i = \{d_j\}_{j\in r}$.
\end{itemize}

Typically, a label has one component indicating the (reduced) cost of the route and several other components recording the consumed resources, i.e., each component has a fixed value. However, our proposed label has a special component $G(r, q)$, namely a function of $q$. Additionally, in our label we do not need a component related to the vehicle capacity. At vertex 0, we define $E_0 = (\tau_0, N_0, V_0, G(r, q) = (F_0, {SL}_0, I_0)) = (0, \{0\}, V_C \cup \{n+1\}, (0, \emptyset, \emptyset))$. Each vertex may have multiple labels and the optimal solution to the pricing subproblem can be achieved by identifying the labels with the smallest $C(r, Q)$ at vertex $n+1$.

A label $E_i$ can be extended to vertex $j \in V^+(i)$, yielding a new label $E_j$. The extension functions are:
\begin{itemize}
\item $\tau_j = \max\{e_j, \tau_i + s_i + t_{i, j}\}$;
\item $N_j = N_i\cup \{j\}$;
\item $V_j = V_i - \{k \in V^+(j): \tau_j + s_j + t_{j, k} > l_k\} - \{j\}$;
\item $F_j = F_i + bc_{i, j} - \mu_j$;
\item ${SL}_j = \{g_k \leftarrow g_k - ac_{i, j}: g_k \in {SL}_i\} \cup \{g_j = \pi_j\}$;
\item $I_j = I_i \cup \{d_j\}$.
\end{itemize}
Note that all labels do not contain any information regarding the order in which the vertices have been visited, and the labels $E_j$ with $V_j = \emptyset$ or $\tau_j > l_j$ are discarded. In the course of the label-setting algorithm, we cyclically examine all vertices, at each of which all labels that do not have successors would be extended. Extending a label at vertex $i$ may create as many new labels as the number of its successors. Undoubtedly, the number of labels would increase exponentially with the extension of the labels. To avoid exhaustive enumeration, dominance rules are employed to identify and eliminate the dominated labels. The performance of the label-setting algorithm heavily depends on the efficiency of the dominance rules, which determine the number of states generated.

The label $E_i$ can be regarded as a set of infinite number of labels $E_i(\hat{q}) = (\tau_i, N_i, V_i, \hat{q}, G(r, \hat{q}))$ for all $0 \leq \hat{q} \leq \sum_{i\in r}d_i$, where the quantity of products collected along the partial route is exactly $\hat{q}$ and $G(r, \hat{q})$ is the associated minimal reduced cost. Let $p_i$ be the partial extreme collection pattern associated with $G(r, \hat{q})$ and $\bar{P}_i$ be the set of all feasible extensions of the label $E_i(\hat{q})$ to vertex $n+1$. We use $p_i\oplus p'$ to denote the complete collection pattern resulting from extending $p_i$ by $p' \in \bar{P}_i$. As stated by \citet{Irnich2005,Desaulniers2010,Dabia2012}, a label $E^1_i(\hat{q}^1)$ is dominated by a label $E^2_i(\hat{q}^2)$ if the following conditions hold:
\begin{itemize}
\item Any feasible extension of $E^1_i(\hat{q}^1)$ is also feasible to $E^2_i(\hat{q}^2)$, namely $\bar{P}_i^1 \subseteq \bar{P}_i^2$;
\item The reduced cost of $p_i^1\oplus p'$ is greater than or equal to that of $p_i^2\oplus p'$ for each $p' \in \bar{P}_i^1$.
\end{itemize}
However, it is not straightforward to verify the above conditions since it requires to evaluate all feasible extensions of both labels. Instead, we propose the following sufficient conditions:
$E^1_i(\hat{q}^1)$ is dominated by $E^2_i(\hat{q}^2)$ if
\begin{enumerate}
\item $\tau_i^2 \leq \tau_i^1$;
\item $V_i^2 \supseteq V_i^1$;
\item $\hat{q}^2 \leq \hat{q}^1$
\item $G(r^2, \hat{q}^2) \leq G(r^1, \hat{q}^1)$
\end{enumerate}

The dominance rule for two labels $E_i(\hat{q}^1)$ and $E_i(\hat{q}^2)$ that have the same route $r$ can be described as: $E_i(\hat{q}^1)$ is dominated by $E_i(\hat{q}^2)$ if conditions 3 and 4 are satisfied. Based on this dominance rule, the labels $E_i(\hat{q})$ associated with the increasing part of the reduced cost function $G(r, q)$ (see Figure \ref{fig:function}) are dominated by label $E_i(q^*)$ and can be safely eliminated. Therefore, we can replace the increasing part of the reduced cost function with a zero slope piece and redefine $G(r, \hat{q})$ as the minimal reduced cost associated with route $r$ and an allowable capacity $\hat{q}$. Subsequently, we derive a dominance rule for two labels $E_i^1$ and $E_i^2$ as: $E^1_i$ is dominated by $E^2_i$ if a label $E^2_i(\hat{q}^2)$ can be always found to dominate $E^1_i(\hat{q}^1)$ for each feasible $\hat{q}^1$. Specifically, the sufficient conditions for $E^2_i$ to dominate $E^1_i$ are: conditions 1 and 2 are satisfied, and  $G^1(r^1, q) \geq G^2(r^2, q)$ for each $q \in [0, Q]$ (see Figure \ref{fig:dom}). Since this dominance rule involves only two labels, we call it the {\em pair dominance rule}.
\begin{figure}[!th]
\begin{center}
\resizebox{8cm}{!}{\includegraphics{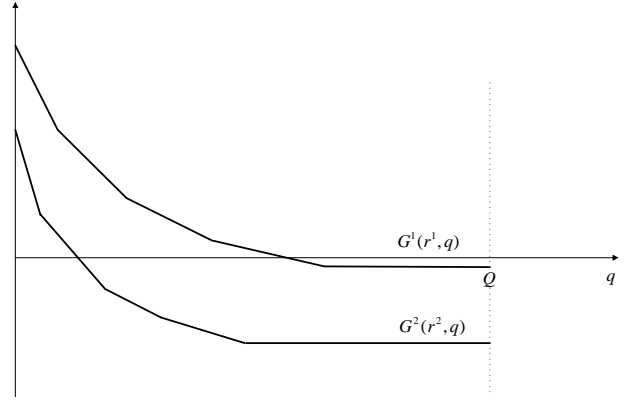}}
\end{center}
\caption{The graphic representations of the functions $G^1(r^1, q)$ and $G^2(r^2, q)$ (note that the increasing parts of both functions have been replaced with zero slope pieces).} \label{fig:dom}
\end{figure}

The pair dominance rule is quite weak since the number of cases that a reduced cost function lies below another one may not be very large. Fortunately, we find that although a label cannot be dominated by another one, it might be dominated by a set of labels. Based on this finding, we introduce a novel dominance rule called the {\em set dominance rule}, which is described as follows. It is easy to derive that $E^1_i$ is a dominated label if a label $E^{x}_i(\hat{q}^{x})$ can always be found to dominate $E^1_i(\hat{q}^1)$ for each feasible $\hat{q}^1$. Denoting by $\mathbb{E}_i$ the set of all labels ending at vertex $i$, we can define a label set $\mathbb{E}_i^1$ related to label $E_i^1$ as: $\mathbb{E}_i^1 =\{E_i^x \in \mathbb{E}_i: \tau_i^x \leq \tau_i^1, V_i^x \supseteq V_i^1, E_i^x \neq E^1_i\}$. Using the labels in $\mathbb{E}_i^1$, we can construct a minimal reduced cost function:
\begin{align}
G_{min}^1(q) = \min_{E_i^x \in  \mathbb{E}_i^1 }\{G^x(r^x, q)\}
\end{align}
As illustrated in Figure \ref{fig:gmin}, the function $G_{min}^1(q)$ is composed of the minimal part of all functions $G^x(r^x, q)$ ($x = 2, 3, 4$), and is not necessarily convex. If the curve $G_{min}^1(q)$ lies below the curve $G^1(r, q)$, we say that $E^1_i$ is dominated by the label set $\mathbb{E}_i^1$ and can be safely discarded. Figure \ref{fig:ex2} gives an example in which $E^1_i$ cannot be dominated by either $E^2_i$ or $E^3_i$, but it is dominated by set $\{E^2_i, E^3_i\}$ according to the set dominance rule. The implementation details of our set dominance rule is described in Appendix A.

\begin{figure}[!th]
\begin{center}
\resizebox{8cm}{!}{\includegraphics{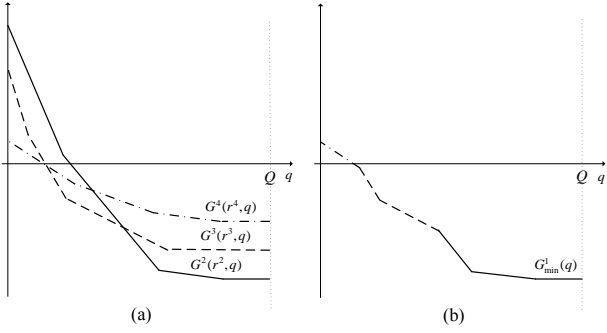}}
\end{center}
\caption{(a) The graphic representations of functions $G^2(r^2, q)$, $G^3(r^3, q)$ and $G^4(r^4, q)$. (b) The graphic representation of function $G^1_{min}(q)$.} \label{fig:gmin}
\end{figure}

\begin{figure}[!th]
\begin{center}
\resizebox{9cm}{!}{\includegraphics{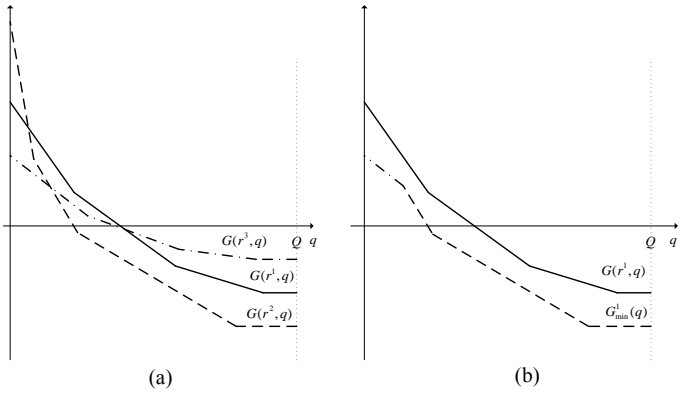}}
\end{center}
\caption{(a) The graphic representations of functions $G^1(r^1, q)$, $G^2(r^2, q)$ and $G^3(r^3, q)$. (b) $G^1_{min}(q)$ lies  below $G^1(r^1, q)$.} \label{fig:ex2}
\end{figure}

The label-setting algorithm proposed by \citet{Desaulniers2010} can be adapted to solve our pricing subproblem. However, compared with that algorithm, our proposed label-setting algorithm exhibits obvious advantages in the following two aspects. First, in the course of our label-setting algorithm, each feasible partial route has at most one label. The label-setting algorithm in \citet{Desaulniers2010} creates a huge number of dominated labels that cannot be efficiently eliminated by their proposed pair dominance rule. When extending a label from vertex $i$ to vertex $j$ ($j \neq n+1$), their label-setting algorithm creates up to three labels, corresponding to zero delivery, full delivery and split delivery, respectively. This type of extension is essentially equivalent to enumerating all feasible extreme delivery patterns compatible with a certain route and as a result a feasible partial route is very likely to be associated with a large number of labels. However, in fact, for each feasible partial route, only at most one label is non-dominated and needs to be kept. Second, our proposed set dominance rule is far more efficient to eliminate the dominated labels since it utilizes the information of all available labels ending at a certain vertex. In brief, the label-setting algorithm proposed by \citet{Desaulniers2010} creates and deals with much more labels than ours, and consequently requires more computational efforts. To show the superiority of our label-setting algorithm, we applied it to solve the LMP for all SDVRPTW instances used in \citet{Desaulniers2010}, and report the experimental results in Section \ref{sec:ce}.

\subsection{Accelerating Strategies}
\label{sub:as}
We have implemented the following three techniques to speed up the column generation procedure.
\subsubsection{Bounded Bidirectional Search}
\label{subsub:bbs}
As discussed in \citet{Righini2006,Righini2008}, the label-setting algorithm can be accelerated by bounded bidirectional search strategy. The resulting algorithm is called the {\em bounded bidirectional label-setting (BBLS) algorithm} whose procedure can be briefly summarized as the following three steps: (1) labels are extended forward from vertex $0$, generating a set of forward partial routes; (2) labels are extended backward from vertex $n+1$, generating a set of backward partial routes; and (3) pairs of forward and backward partial routes are joined together to generate complete routes.

The BBLS algorithms have  been successfully employed to solve the pricing subproblems for a variety of vehicle routing models, such as the VRPTW \citep{Desaulniers2008}, the VRP with simultaneous distribution and collection \citep{Dell2006}, the pickup and delivery problem with time windows \citep{Ropke2009}, the SDVRPTW \citep{Desaulniers2010} and the VRPTW with multiple use of vehicles \citep{Azi2010}. In the BBLS algorithms developed in the above-mentioned articles, the forward and backward extensions are almost the same due to the symmetric structure of the investigated problems. We find that a common objective of these problems is to minimize the overall traveling distance of all vehicles. This feature results in a property that the costs of backward and forward partial routes are independent, i.e., the cost of a backward partial route does not rely on its forward partial route.

If the route cost is determined by the arrival time at each vertex or the flow on each edge, the symmetric structure of the vehicle routing models would be destroyed to a certain extent; we call this type of cost the {\em cumulative cost}. In the vehicle routing models with cumulative costs, e.g., the SCVRPTWL, we can find that the cost of a backward partial route is heavily relied on its forward partial route. In \citet{Mattos2012}, the authors developed a branch-price-and-cut algorithm to solve the workover rig routing problem (WRRP) that incorporates a cumulative cost at each vertex. They claimed that the bounded bidirectional search strategy cannot be applied to their pricing subproblem. We could not make a conclusion on whether there exists a BBLS algorithm for the WRRP. However, after carefully analyzing the structure of our pricing subproblem, we find that it can still be optimally solved by a tailored BBLS algorithm in which the forward and backward extensions are considerably different.

The process of the bidirectional search strategy is pictorially shown in Figure \ref{fig:back}. Given a backward partial route $r = (v(1), v(2), \ldots, v(|r|))$, where $|r|\geq 2$ and $v(|r|) = n+1$, and the incoming flow $\hat{q}$, the minimal reduced cost $G^b(r, Q - \hat{q})$ of all backward partial collection patterns can be computed by:
\begin{align}
G^b(r, Q - \hat{q})=&f_r^b - \sum_{i = 1}^{|r|}\delta_{v(i)}g_{v(i)}^b\\
\mbox{s.t.}~~ &\sum_{i = 1}^{|r|}\delta_{v(i)} = Q - \hat{q} \nonumber\\
&0\leq \delta_{v(i)} \leq d_{v(i)}, ~\forall~ 1 \leq  i \leq |r| \nonumber
\end{align}
where
\begin{align}
f_r^b =& \sum_{i = 1}^{|r| -1}\Bigg{(}(aQ + b)c_{v(i), v(i+1)} - \mu_{v(i)}\Bigg{)};\nonumber\\
g_{v(1)}^b =& \pi_{v(1)}; \nonumber\\
g_{v(i)}^b = & a\sum_{j = 1}^{i -1}c_{v(j), v(j+1)} + \pi_{v(i)}, ~\forall~ 2 \leq i \leq |r|-1 ;\nonumber\\
g_{|r|}^b = & a\sum_{j = 1}^{|r|-1}c_{v(j), v(j+1)}.\nonumber
\end{align}
We refer the reader to Appendix B for the detailed derivation of $G^b(r, Q - \hat{q})$. Algorithm \ref{alg:greedy} can still be used to compute the value of $G^b(r, Q - \hat{q})$. When route $r$ is fixed, $G^b(r, Q - q)$ can be viewed as a function of the allowable capacity $Q - q $ (e.g., see Figure \ref{fig:function2}).

\begin{figure}[!th]
\begin{center}
\resizebox{8cm}{!}{\includegraphics{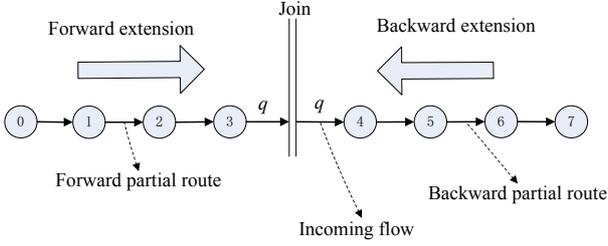}}
\end{center}
\caption{ The bidirectional search strategy.} \label{fig:back}
\end{figure}

\begin{figure}[!th]
\begin{center}
\resizebox{7cm}{!}{\includegraphics{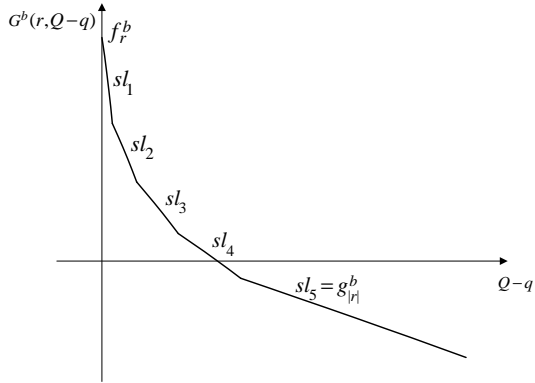}}
\end{center}
\caption{Graphic representation of the reduced cost function $G^b(r, Q - q)$.} \label{fig:function2}
\end{figure}

We use a label $E^b_i = \Big{(}\tau^b_i, N^b_i, V^b_i, G^b(r, Q - q) = (F^b_i, SL_i^b, I_i^b)\Big{)}$ to represent a state in backward extension, where:
\begin{itemize}
\item $\tau^b_i$ represents the minimum time which must be consumed since the departure from vertex $i$ up to the arrival at vertex $n+1$;
\item $N_i^b$ and $V^b_i$ have the same definitions as $N_i$ and $V_i$;
\item $G^b(r, Q - q)$ represents the reduced cost function associated with the backward partial route $r$ and the allowable capacity $Q - q$, which could be represented by: $F^b_i = f^b_r$, ${SL}_i^b = \{g^b_j\}_{j \in r}$ and $I_i^b=\{d_j\}_{j \in r}$.
\end{itemize}
We define $E_{n+1}^b$ $= \big{(}\tau_{n+1}^b,$ $N_{n+1}^b,$ $V_{n+1}^b,$ $G^b(r, Q - q)$ $= (F_{n+1}^b,$ ${SL}^b_{n+1},$ $I_{n+1}^b)\big{)}$ $= \big{(}0, \{n+1\},$ $V_C\cup \{0\},$ $(0, \{g^b_{n+1}\},$ $\{d_{n+1}\})\big{)}$. The overall time resource $T$ is equal to the maximum feasible arrival time at vertex $n+1$, namely $T = \max_{i\in V_C}\{l_i + s_i + t_{i, n+1}\}$. Analogously, the backward partial route is extended from vertex $j$ to vertex $i$ according to the following functions:
\begin{itemize}
\item $\tau_i^b = \max\{T - l_i - s_i , \tau_j^b + s_j + t_{i, j}\}$;
\item $N^b_i = N^b_j \cup \{i\}$;
\item $V^b_i = V^b_j - \Big{\{}k \in V^-(i): \tau_i^b + s_i + t_{k, i} > T - e_k - s_k \Big{\}} - \{i\}$;
\item $F^b_i = F^b_j + (aQ + b)c_{i, j} - \mu_i$;
\item ${SL}_i^b =\Big{\{}g_k \leftarrow g_k + ac_{i, j}: g_k \in {SL}^b_j\Big{\}} \cup \{g_i = \pi_i\}$;
\item $I^b_i = I^b_i \cup \{d_i\}$.
\end{itemize}
We discard the labels with $V^b_i = \emptyset$ or $\tau_i > T - e_i - s_i$, and still use the set dominance rule that is described in Section \ref{sub:lsa} to eliminate the dominated backward partial routes.

When applying the BBLS algorithm, we consider time as the critical resource and only extend forward and backward labels whose consumed time resources are less than $T/2$, namely $\tau_i < T/2$ and $\tau^b_i < T/2$.  A forward label $E_i = \big{(}\tau_i, N_i, V_i, G(r, q)\big{)}$ and a backward label $E^b_j = \big{(}\tau^b_j, N^b_j, V_j^b, G^b(r, Q- q)\big{)}$ can be joined together to form a complete feasible route if $\tau_i + s_i + t_{i, j} + \tau^b_j \leq T$ and $N_i \cap N_j^b = \emptyset$. The cost of the resulting complete collection pattern is achieved using the information of $G(r, q)$ and $G^b(r, Q - q)$ as follows. The fixed cost of the complete collection patten is the sum of $F_i + F_j^b + bc_{i, j}$. With the values of $g_k \in {{SL}_i \cup {SL}_j^b}$ and the collected quantity $Q$, we can use Algorithm \ref{alg:greedy} to decide the quantity $\delta_k$ collected at each visited vertex $k$ and then compute the variable cost of the complete collection pattern as $- \sum_{k \in N_i \cup N_j^b}\delta_kg_k$. The minimum cost among all complete collection patterns is the optimal solution value of the pricing subproblem. Usually, at each column generation iteration we identify a number of columns with negative reduced cost and then add them into the current RLMP.

\subsubsection{Heuristic Column Generator}
Heuristics may identify negative reduced cost columns with much less computation time, compared to the exact label-setting algorithm. To avoid solving the pricing subproblem optimally at each column generation iteration, we develop an adaptive greedy heuristic (AGH), as shown in Algorithm \ref{alg:agh}, to heuristically and rapidly identify negative reduced cost columns. At each column generation iteration, we first use the AGH to solve the pricing subproblem. If it manages to obtain some columns with negative reduced cost, we add these columns into the RLMP and start the next iteration. Otherwise, we invoke the BBLS algorithm to solve the pricing subproblem to optimality.

The AGH tries to identify up to {\em maxCol} negative reduced cost columns, starting with a set $R_0$ of routes with zero reduced cost in the optimal solution of the current RLMP. We define {\em maxIter} as the maximum number of iterations associated with each route in $R_0$, $\rho_i$ as a {\em valuation} for each vertex $i$ that is used to calculate its priority value, and $\eta$ ($0 < \eta < 1$) as a penalty factor. Since the best extreme collection pattern of a given route can be easily obtained, in the AGH we use a route to represent a solution of the pricing subproblem.  At the beginning of the outer loop, the valuations of all vertices are initialized to one and all vertices that are not included in the current route $r$ are stored in a vertex queue {\em vertex\_queue} in an order of decreasing value of $\pi_i \times \rho_i$ (lines \ref{line:agh:01} -- \ref{line:agh:02} in Algorithm \ref{alg:agh}).

In each iteration, the heuristic pops the vertices $u$ one by one from {\em vertex\_queue}, checks their insertion into the current route $r$ by a subroutine {\em GreedyInsert}($r$, $u$) given in Algorithm \ref{alg:greedyinsert}, and adds the resulting feasible routes with negative reduced cost into a route pool {\em route\_pool} (see lines \ref{line:agh:1} -- \ref{line:agh:2} in Algorithm \ref{alg:agh}). Upon collecting {\em maxCol} negative reduced cost columns, we terminate the heuristic and return {\em rout\_pool} (lines \ref{line:agh:1.1} -- \ref{line:agh:1.2} in Algorithm \ref{alg:agh}). If the resulting route $r'$ is better than the currently best route $r^*$, we update $r^*$ by $r'$. Otherwise, the corresponding parameter $\rho_u$ is decreased to $\eta\times \rho_u$, which delays the checking of its insertion in the next iteration. This strategy is very similar to the adaptive process used by other search techniques such as the ejection pool algorithm \citep{Luo2013}. Whenever {\em vertex\_queue} becomes empty, the heuristic removes one non-depot vertex $u$ from the current route by alternatively using a greedy procedure {\em GreedyRemove}($r$) shown in Algorithm \ref{alg:greedyremove} and a random way, decrease $\rho_u$ to $\eta\times \rho_u$, and reset {\em vertex\_queue} based on toggling between two sorting rules (see lines \ref{line:agh:3} -- \ref{line:agh:4} in Algorithm \ref{alg:agh}). To obtain the columns, we compute the best collection patterns compatible with each route $u$ in {\em route\_pool} by performing Algorithm \ref{alg:greedy} with consideration of only the vertices with $g_{v(i)} > 0$.

\begin{algorithm}[!h]
\caption{The Adaptive Greedy Heuristic.} \label{alg:agh}
\begin{tiny}
\begin{algorithmic}[1]
\STATE INPUT: A set of routes $R_0$, {\em maxCol}, {\em maxIter} and $\eta$;
\STATE Define {\em vertex\_queue} and  {\em route\_pool} as a vertex queue and a route pool, respectively;
\STATE Set {\em flag} $\leftarrow$ false, $k  \leftarrow $ {\em maxIter} and $r^* \leftarrow$ any route $r \in R_0$;
\WHILE{$k \leq $ {\em maxIter} and $R_0$ is not empty}
\IF{$k$ = {\em maxIter} }
\STATE $k \leftarrow 1$;
\STATE $r$ $\leftarrow$ randomly select one route from $R_0$ and remove $r$ from $R_0$;
\STATE Set $\rho_i \leftarrow 1$ for all $i \in V$ and {\em vertex\_queue} $\leftarrow$ $V_C - $ \{all vertices in $r$\}; \label{line:agh:01}
\STATE Sort all vertices $i$ in {\em vertex\_queue} by decreasing value of $\pi_i \times \rho_i$;\label{line:agh:02}
\ENDIF
\WHILE{{\em vertex\_queue} is not empty}\label{line:agh:1}
\STATE $u \leftarrow$ pop the top element in {\em vertex\_queue};
\STATE $r' \leftarrow$ {\em GreedyInsert}($r$, $u$);
\IF{$r' \neq$ {\em null} and $C(r', Q) < C(r*, Q)$}
\STATE $r \leftarrow r'$;
\ENDIF
\IF{$r' \neq$ {\em null} and $C(r', Q)$ is negative}
\STATE Add $r'$ into {\em route\_pool};
\IF{the size of {\em route\_pool} is equal to {\em maxCol}} \label{line:agh:1.1}
\RETURN {\em route\_pool};
\ENDIF \label{line:agh:1.2}
\IF{$C(r', Q) < C(r^*, Q)$}
\STATE $r^* \leftarrow r'$;
\STATE $\rho_u \leftarrow {\rho_u}/{\eta}$;
\ENDIF
\ENDIF
\ENDWHILE \label{line:agh:2}
\IF{{\em flag} $=$ false}\label{line:agh:3}
\STATE $r$ $\leftarrow$ {\em GreedyRemove}($r$);
\STATE {\em vertex\_queue} $\leftarrow$ $V_C - $ \{all vertices in $r$\};
\STATE Sort all vertices $i$ in {\em vertex\_queue} by decreasing value of $d_i \times \pi_i \times \rho_i$;
\STATE {\em flag } $\leftarrow$ true;
\ELSE
\STATE $r$ $\leftarrow$ randomly remove a vertex $u$ except $0$ and $n+1$ from $r$ and set $\rho_u \leftarrow \eta \times \rho_u$;
\STATE {\em vertex\_queue} $\leftarrow$ $V_C - $ \{all vertices in $r$\};
\STATE Sort all vertices $i$ in {\em vertex\_queue} by decreasing value of $\pi_i \times \rho_i$;
\STATE {\em flag } $\leftarrow$ false;
\ENDIF \label{line:agh:4}
\STATE $k \leftarrow k +1$;
\ENDWHILE
\RETURN {\em route\_pool};
\end{algorithmic}
\end{tiny}
\end{algorithm}

\begin{algorithm}[!h]
\caption{{\em GreedyInsert}($r$, $i$).} \label{alg:greedyinsert}
\begin{small}
\begin{algorithmic}[1]
\STATE $r' \leftarrow $ {\em null};
\FOR{each pair of two consecutive vertices $u$ and $v$ in $r$}
\STATE Insert $i$ between $u$ and $v$;
\IF{the resulting $r$ is infeasible}
\STATE {\bf Continue};
\ELSIF{$r'$ is not initialized}
\STATE $r' \leftarrow r$;
\ELSIF{$C(Q, r) < C(Q, r')$}
\STATE $r' \leftarrow r$;
\ENDIF
\STATE Restore $r$ to its state before inserting $i$;
\ENDFOR
\RETURN $r'$.
\end{algorithmic}
\end{small}
\end{algorithm}

\begin{algorithm}[!h]
\caption{{\em GreedyRemove}($r$).} \label{alg:greedyremove}
\begin{small}
\begin{algorithmic}[1]
\STATE $r' \leftarrow$ {\em null};
\FOR{each vertex $u \in r$ except $0$ and $n+1$}
\STATE Remove $u$ from $r$;
\IF{$r'$ is not initialized}
\STATE $r' \leftarrow r$;
\ELSIF{$C(Q, r) < C(Q, r')$}
\STATE $r' \leftarrow r$ and $v \leftarrow u$;
\ENDIF
\STATE Restore $r$ to its state before deleting $u$;
\ENDFOR
\STATE $\rho_v \leftarrow \eta \times \rho_v$;
\RETURN $r'$.
\end{algorithmic}
\end{small}
\end{algorithm}

\subsubsection{Decremental Search Space}
The decremental search space was introduced independently by \citet{Boland2006} and \citet{Righini2008}. It starts from solving the pricing subproblem with the elementary requirements of all customers being relaxed, i.e., each customer can be visited more than once in a route. In our label-setting algorithm, if the elementary requirement of vertex $j$ is relaxed, it will not be removed from $V_j$ when the label is extended from vertex $i$ to vertex $j$, and is allowed to exist in $N_i\cap N_j^b$ when joining labels. If the computed least-cost path is nonelementary, the customers that are visited more than once are required to be elementary and the pricing subproblem is solved again. This process is repeated until an elementary least-cost route is found. Our implementation of the decremental search space technique is the same as the one described in \citet{Desaulniers2010}. This acceleration technique has also been employed in the branch-and-price algorithms for solving several other vehicle routing models, such as the VRP with discrete split deliveries and time windows \citep{Salani2011}, the VRP with deliveries, selective pickups and time windows \citep{Gutierrez-Jarpa2010} and the multi-depot VRPTW \citep{Bettinelli2011}.

\section{Branch-and-Price-and-Cut Algorithm}
\label{sec:brc}
Branch-and-price-and-cut is one of the leading solution procedure for many large-scale integer programming models (e.g., see \citet{Ropke2009,Barnhart2000,Belov2006,Hwang2008}). Over the course of the branch-and-bound search, some violated valid inequalities are dynamically added into the model. In our branch-and-price-and-cut algorithm, the initial set of columns corresponds to the set of all one-customer routes, namely $r = (0, i, n+1)$ for each $i \in V_C$. At each branch-and-bound node, we first optimally solve the LMP using the column generation procedure to obtain a lower bound. For the node that cannot be pruned, we next try to identify the $k$-path inequalities and {\em strong minimum number of vehicles inequalities} that are violated by the current linear solution. If such violated inequalities are found, we add them into the model and invoke the column generation procedure again to further improve the lower bound. The above procedure is repeated until the node is pruned or no violated inequalities can be found.

In this section, we first describe two types of valid inequalities. This is followed by search and branching strategies that guide the exploration of the branch-and-bound tree.

\subsection{Valid Inequalities}
We use two types of valid inequalities for the SCVRPTWL, namely the $k$-path inequality and the {\em strong minimum number of vehicles inequalities}, which have been implemented by \citet{Archetti2011} for the SDVRPTW. These inequalities are defined on the master problem variables $\theta_{r,p}$. After adding some valid inequalities into the master problem, the subproblem as well as the label-setting algorithm need to be modified accordingly. Below, we only discuss in detail the treatment of these inequalities in forward extension. The modifications on backward extension can be easily derived in a similar manner.

\subsubsection{$k$-path Inequalities}
The $k$-path inequalities are expressed as:
\begin{align}
&\sum_{r \in R}\sum_{p \in P_r}\sum_{(i, j) \in E^-(S)} \beta_{i, j, r} \theta_{r, p} \geq \Bigg{\lceil} \frac{\sum_{i \in S}d_i}{Q} \Bigg{\rceil}, ~\forall~ S \in \Gamma \label{eqn:mipmp:c3}
\end{align}
where the binary parameter $\beta_{i, j, r} = 1$ if edge $(i, j)$ is used in route $r$, $E^-(S) = \{(i, j)\in E| i \in S, j \notin S\}$ is the set of edges leaving the customer subset $S$, and $\Gamma$ is the set of the subsets $S \in V_C$. Let $\lambda = (\lambda_{S_1}, \ldots, \lambda_{S_{|\Gamma|}})$ be the values of the dual variables associated with Constraints (\ref{eqn:mipmp:c3}). The reduced cost $\bar{c}_{r, p}$ and the fixed cost $f_r$ become:
\begin{align}
\bar{c}_{r, p} &=  c_{r, p} - \sum_{i \in r}(\delta_{i, p} \pi_i  + \mu_i) - \sum_{S\in \Gamma} \sum_{(i, j)\in E^-(S)\cap r}\lambda_S \nonumber\\
f_r = &\sum_{i = 1}^{|r|-1} \bigg{(}bc_{v(i), v(i+1)}  -  \sum_{S\in \Gamma
:(v(i), v(i+1))\in E^-(S)} \lambda_S \bigg{)}  - \sum_{i = 1}^{|r|}\mu_{v(i)} \nonumber
\end{align}

Handling the new dual variable $\lambda_S$ in the label-setting algorithm needs to modify the extensions function related to the fixed cost as follows:
\begin{align}
F_j &= F_i + bc_{i, j} - \mu_j - \sum_{S \in \Gamma: (i, j) \in E^-(S)} \lambda_S \nonumber
\end{align}
Moreover, when joining two labels, the fixed cost of the complete collection patten becomes $F_i + F_j^b + bc_{i, j} - \sum_{S \in \Gamma: (i, j) \in E^-(S)} \lambda_S$. It is worthy to mention that the introduction of the $k$-path inequalities only affects the fixed cost of the reduced cost.

To identify the violated $k$-path inequalities, we have implemented three types of separation heuristics, which have been used in the branch-and-price-and-cut algorithms for the SDVRPTW \citep{Desaulniers2010,Archetti2011}. The first one was the {\em partial enumeration heuristic} proposed by \citet{Desaulniers2010} and the other two were the {\em extended shrinking heuristic} and the {\em route-based algorithm} developed by \citet{Archetti2011}. We refer the reader to these two articles for full details of these three separation heuristics. Note that our $k$-path inequalities only take the vehicle capacity constraints into consideration.

\subsubsection{Strong Minimum Number of Vehicles (SMV) Inequalities}
Define $V_S$ as the set of customers $i \in V_C$ with $d_i \leq Q$. The SMV inequalities are expressed as:
\begin{align}
\sum_{r\in R}\sum_{p \in P_r} (2\alpha_{i, r, p}^F + \alpha_{i, r, p}^{SZ})\theta_{r, p} \geq 2. ~\forall~i \in V_S \label{eqn:mipmp:smv}
\end{align}
where at customer $i$ in pattern $p$ compatible with route $r$, if a full collection is performed, then the binary parameter $\alpha_{i, r, p}^F =1$, and if a split or zero collection is performed, then the binary parameter $\alpha_{i, r, p}^{SZ} =1$. This type of inequality was first proposed by \citet{Archetti2011}. Let $\gamma = (\gamma_1, \ldots, \gamma_n)$ be the values of the dual variables associated with Constraints (\ref{eqn:mipmp:smv}). The reduced cost $\bar{c}_{r, p}$ becomes
\begin{align}
\bar{c}_{r, p} &=  c_{r, p} - \sum_{i \in r}(\delta_{i, p} \pi_i  + \mu_i) - \sum_{i \in V_S}(2\alpha_{i, r, p}^F + \alpha_{i, r, p}^{SZ})\gamma_i \nonumber\\
\end{align}

To deal with the SMV inequalities, we need to modify the label-setting algorithm. First, define a new label $E_i$ that contain additional components as follows:
\begin{align}
E_i = \Big{(} {\tau}_i, N_i, V_i, G(r, q) = (F_i, {SL}_i, I_i), (\chi_i^j)_{j \in V_S}, Q_i \Big{)} \nonumber
\end{align}
where $\chi_i^j$ for all $j \in V_S$ are initialized to zero, $\chi_i^j = 1$ indicates that customer $j$ is forced to be full collection and $Q_i$ is the remaining capacity. Next, if $\gamma_j > 0$, we need to create two types of labels along edge $(i, j)$: type 1 label is for a zero or split collected is performed and type 2 label is for a full collection. The forward extension functions involving the new dual variables $\gamma_i$ and the new label components are as follows. For type 1 label, we have:
\begin{align}
F_j &= F_i + bc_{i, j} - \mu_j - \gamma_i;\nonumber \\
\chi_j^k& = \chi_i^k; \nonumber\\
Q_j &= Q_i.\nonumber
\end{align}
and for type 2 label, we have:
\begin{align}
F_j &= F_i + bc_{i, j} - \mu_j - 2\gamma_i;\nonumber\\
\chi_j^k & = \left\{
\begin{array}{ll}
1, & \textrm{if $k = j$;}\\
\chi_j^k, & \textrm{otherwise;}\\
\end{array}
\right. \nonumber\\
Q_j & = Q_i - d_j.\nonumber
\end{align}

Given a label $E_i$, the minimum reduced cost is calculated as follows. The fixed cost $f_r$ is the sum of $F_i$ and $- \sum_{k\in V_S: \chi_i^k =1}d_kg_k$. With the values of $g_k \in {SL}_i -  \{v: \chi_i^v = 1\}$ and $\hat{q}$ equal to the remaining capacity $Q_i$, we can perform the modified Algorithm \ref{alg:greedy} to achieve the value of $-\sum_{k \in N_i - \{v: \chi_i^v = 1\}}\delta_kg_k$. The graph $G(r, q)$ is illustrated in Figure \ref{fig:smv}, where we assume $G(r, q)$ with $0 \leq q \leq \sum_{k\in V_S: \chi_i^k =1}d_k$ equals a sufficiently large positive constant $M$.

\begin{figure}[!th]
\begin{center}
\resizebox{8cm}{!}{\includegraphics{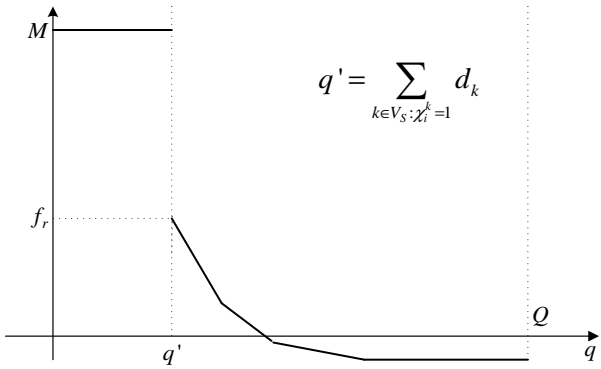}}
\end{center}
\caption{The graphic representation of $G(r, q)$ after the introduction of SMV inequalities.} \label{fig:smv}
\end{figure}

\subsection{Search strategy}
\label{sub:ss}
The branch-and-bound tree is explored according to a best-first policy; specifically, the ``best'' unexamined tree node is the one with the smallest lower bound, and would be given the highest priority. We have tested the depth-first policy in some preliminary experiments and obtained inferior results in terms of the number of the optimally solved instances within the same amount of computation time.

\subsection{Branching strategies}
\label{sub:bs}
At each branch-and-bound node, we achieve an optimal solution of the LMP using column generation procedure and separation heuristics; this solution value is a lower bound at that node. If this lower bound is not less than the current upper bound, the associated node is pruned; otherwise, branching must take place. If the optimal solution of the LMP is integral and the optimal solution value is less than the current upper bound, we update the upper bound.

As explained in \citet{Desaulniers1998}, we can hardly branch on master problem variables $\theta_{r,p}$ since fixing such variables at 0 requires preventing label-setting algorithms from generating the corresponding routes, significantly increasing the complexity of solving the pricing subproblem. Therefore, it is better to choose branching strategies compatible with the algorithms for the pricing subproblems, i.e., the pricing subproblems at the nodes resulting from such branchings could be solved in a way similar to the one used at their parent nodes. This requires that branching constraints do not change the structure of the pricing subproblem. In our branch-and-price-and-cut algorithms, we choose four types of branching strategies that have been implemented in \citet{Desaulniers2010}, namely branching on the total number of vehicles used, on the number of vehicles visiting each customer, on the total flow on each edge and on including or not including two consecutive edges in the vehicle routes.

\section{Computational Experiments}
\label{sec:ce}

\subsection{Instances}
To evaluate the branch-and-price-and-cut algorithm proposed in this paper, we conducted experiments using the data set derived from the 56 benchmark VRPTW instances of \citet{Solomon1987}, which are divided into six groups, namely R1, C1, RC1, R2, C2 and RC2. Each Solomon instance contains a designated depot and 100 customers, for a total of 101 vertices. From these 100-customer instances, we derived 25-customer and 50 customer instances by only considering the first 25 and 50 instances, respectively. For each of these instances, we consider three type of vehicle capacities, namely $Q = 30, 50$ and 100. Thus, we have 504 instances in total, which consists of 54 groups. Each group is identified by three parts separated by dashes (`--'), i.e., the name of Solomon group, the number of customers ($n$) and the vehicle capacity ($Q$). For example, instance group R1-100-30 contains all instances with $n = 100$ and $Q = 30$ generated from Solomon instance group R1, and R101-100-30 is the identifier of the first instance in this group. As did by \citet{Desaulniers2010} and \citet{Archetti2011}, the Euclidean distance between any pair of vertices was rounded off to one decimal place. We set $a =1$ and $b = Q/4$ for the weight-related cost function $f(w) = a \times w + b$. This implies that one dollar is charged per unit distance per unit weight and the vehicle weight equals one-quarter of the weight that a vehicle can carry. Note that if we set $a = 0$ and $b = 1$, the instances becomes SDVRPTW instances used in \citet{Gendreau2006,Desaulniers2010,Archetti2011}.
All instances as well as detailed experimental results are available in the online supplement to this paper at: \url{www.computational-logistics.org/orlib/scvrptwl}.

\subsection{Experimental Setup}
 Our algorithm was coded in Java and all experiments were conducted on a Dell server with an Intel Xeon E5520 2.26 GHz CPU, 8 GB RAM and Linux operating system. The linear programming models were solved by simplex algorithm implemented by ILOG CPLEX 12.0.
Computation times reported are in CPU seconds on this server.

We imposed a time limit of 3,600 seconds on each execution of the branch-and-price-and-cut algorithm. However, when the time limit is reached, we do not terminate the algorithm until it finishes processing the current branch-and-bound node. The parameters used in our heuristic column generator were fixed as: {\em maxCol} = 1000, {\em maxIter} = 25 $\times n$ and $\eta = 0.15$.

\subsection{Results on the SDVRPTW instances}
First, we applied our branch-and-price-and-cut algorithm to solve the SDVRPTW instances. At the root node of the branch-and-bound tree, we solved the linear relaxation of the problem without introducing any valid inequality using the column generation procedure. The linear relaxation results are mainly determined by the performance of the label-setting algorithm. We compare our results with the results taken from \citet{Desaulniers2010} and \citet{Archetti2011} in Table \ref{tab:1}.

The results produced by our algorithm are presented in the blocks ``New''; the column ``\# inst'' gives the number of instances in the Solomon instance group; the columns ``\# solved'' give the numbers of optimally solved instances within the time limit; and the columns ``Time'' show the average computation times. Since the experimental environments (Language C/C++, a Linux PC equipped with a Pentium D processor clocked at 2.8 GHz and CPLEX 10.1.1) used by \citet{Desaulniers2010,Archetti2011} are quite different from ours, we can not directly judge whether our results are  better than theirs. We can only say that we have achieved the optimal solutions to the linear relaxations of all SDVRPTW instances used by \citet{Desaulniers2010,Archetti2011}.

\begin{table*}[!htp]
\normalsize{\caption{Linear Relaxation Results on the SDVRPTW Instances.}
\label{tab:1}}
\begin{center}
\scalebox{0.4}{
\begin{tabular}{ccccccccccccccccccccccc}
\hline
           &            &            &                                                   \multicolumn{ 6}{c}{$Q=30$} &            &                                                   \multicolumn{ 6}{c}{$Q=50$} &            &                                                  \multicolumn{ 6}{c}{$Q=100$} \\
\cline{4 - 9} \cline{11 - 16} \cline{18 - 23}
           &    Solomon &            & \multicolumn{ 2}{c}{New } & \multicolumn{ 2}{c}{Desaulniers (2010)} & \multicolumn{ 2}{c}{Archetti et al. (2011)} &            & \multicolumn{ 2}{c}{New } & \multicolumn{ 2}{c}{Desaulniers (2010)} & \multicolumn{ 2}{c}{Archetti et al. (2011)} &            & \multicolumn{ 2}{c}{New} & \multicolumn{ 2}{c}{Desaulniers (2010)} & \multicolumn{ 2}{c}{Archetti et al. (2011)} \\
\cline{4 - 9} \cline{11 - 16} \cline{18 - 23}
         $n$ &      group &    \# inst & \# solved & Time &  \# solved &       Time &  \# solved &       Time &            &  \# solved &       Time &  \# solved &       Time &  \# solved &       Time &            &  \# solved &       Time &  \# solved &       Time &  \# solved &       Time \\
\hline
        25 &         R1 &         12 &         12 &         $<1$ &         12 &         $<1$ &         12 &         $<1$ &            &         12 &         $<1$ &         12 &         1  &         12 &         $<1$ &            &         12 &         $<1$ &         12 &         1  &         12 &         $<1$ \\

           &         C1 &          9 &          9 &         $<1$ &          9 &         $<1$ &          9 &         $<1$ &            &          9 &          1 &          9 &         1  &          9 &         $<1$ &            &          9 &          1 &          9 &         2  &          9 &         2  \\

           &        RC1 &          8 &          8 &          &          8 &         $<1$ &          8 &        $<1$ &            &          8 &          1 &          8 &         $<1$ &          8 &         $<1$ &            &          8 &          1 &          8 &         1  &          8 &         $<1$ \\

        50 &         R1 &         12 &         12 &          2 &         12 &         2  &         12 &          1 &            &         12 &          3 &         12 &         7  &         12 &          4 &            &         12 &          3 &         12 &        15  &         12 &         8  \\

           &         C1 &          9 &          9 &          2 &          9 &         2  &          9 &          1 &            &          9 &          3 &          9 &         5  &          9 &          3 &            &          9 &          4 &          9 &        21  &          9 &         9  \\

           &        RC1 &          8 &          8 &          2 &          8 &         1  &          8 &         $<1$ &            &          8 &          3 &          8 &         2  &          8 &          2 &            &          8 &          4 &          8 &        10  &          8 &         6  \\

       100 &         R1 &         12 &         12 &         19 &         12 &        31  &         12 &         20 &            &         12 &         23 &         12 &       118  &         12 &         78 &            &         12 &         44 &         12 &       597  &         12 &       384  \\

           &         C1 &          9 &          9 &         14 &          9 &         9  &          9 &          7 &            &          9 &         15 &          9 &        33  &          9 &         25 &            &          9 &         18 &          9 &       130  &          9 &        70  \\

           &        RC1 &          8 &          8 &         16 &          8 &        11  &          8 &         10 &            &          8 &         23 &          8 &        48  &          8 &         39 &            &          8 &         36 &          8 &       325  &          8 &       215  \\

        25 &         R2 &         11 &         11 &          1 &         11 &         2  &         11 &         $<1$ &            &         11 &          2 &         11 &         9  &         11 &          1 &            &         11 &          2 &         11 &        75  &         11 &         2  \\

           &         C2 &          8 &          8 &          1 &          8 &         2  &          8 &         $<1$ &            &          8 &          1 &          8 &         2  &          8 &          1 &            &          8 &          2 &          8 &         9  &          8 &         4  \\

           &        RC2 &          8 &          8 &         <1 &          8 &         1  &          8 &        $<1$ &            &          8 &          1 &          8 &         1  &          8 &         <1 &            &          8 &          2 &          8 &        14  &          8 &         1  \\

        50 &         R2 &         11 &         11 &          3 &         11 &        51  &         11 &          4 &            &         11 &          5 &         11 &       107  &         11 &         15 &            &         11 &         10 &         11 &       896  &         11 &        64  \\

           &         C2 &          8 &          8 &          3 &          8 &         7  &          8 &          2 &            &          8 &          4 &          8 &        17  &          8 &          9 &            &          8 &          7 &          8 &       161  &          8 &        46  \\

           &        RC2 &          8 &          8 &          2 &          8 &         9  &          8 &          1 &            &          8 &          3 &          8 &         8  &          8 &          5 &            &          8 &          6 &          8 &       245  &          8 &        24  \\

       100 &         R2 &         11 &         11 &         22 &         11 &     1,044  &         11 &         54 &            &         11 &         32 &         11 &     1,045  &         11 &        245 &            &         11 &        204 &          5 &       189  &          7 &       869  \\

           &         C2 &          8 &          8 &         19 &          8 &        44  &          8 &         17 &            &          8 &         23 &          8 &       127  &          8 &         91 &            &          8 &         46 &          8 &       812  &          8 &       382  \\

           &        RC2 &          8 &          8 &         18 &          8 &       264  &          8 &         36 &            &          8 &         30 &          8 &       418  &          8 &        156 &            &          8 &        101 &          6 &     1,037  &          8 &     1,130  \\
\hline
\end{tabular}
}
\end{center}
\end{table*}

From Table \ref{tab:1}, we find that the algorithm proposed in \citet{Archetti2011} failed to optimally solve the linear relaxation of four instances in group R2-100-100. However, this article does not reveal the names of these four instances. So we present in Table \ref{tab:2} the optimal linear relaxation values (LP) and associated computation times (LP time) of all instances in group R2-100-100.

\begin{table}[!htp]
\normalsize{\caption{Linear Relaxation Results on the SDVRPTW Instances in Group R2-100-100.}
\label{tab:2}}
\begin{center}
\scalebox{0.9}{
\begin{tabular}{ccccc}
\hline
  Instance & $n$ &  $Q$ & LP  &    LP time \\
\hline
      R201 &        100 &        100 &     734.4  &      1.1  \\

      R202 &        100 &        100 &     726.2  &      1.1  \\

      R203 &        100 &        100 &     723.0  &      1.3  \\

      R204 &        100 &        100 &     723.0  &      1.3  \\

      R205 &        100 &        100 &     728.6  &      1.1  \\

      R206 &        100 &        100 &     725.8  &      1.1  \\

      R207 &        100 &        100 &     723.0  &      1.3  \\

      R208 &        100 &        100 &     723.0  &      1.3  \\

      R209 &        100 &        100 &     723.0  &      1.1  \\

      R210 &        100 &        100 &     727.9  &      1.1  \\

      R211 &        100 &        100 &     723.0  &      1.4  \\
\hline
\end{tabular}
}
\end{center}
\end{table}

During the experiments, we found that some instances violate the triangle inequality due to rounding the distance to one decimal place. Consequently, the optimal solutions reported in \citet{Desaulniers2010,Archetti2011} for some SDVRPTW instances are not truly optimal. To resolve this issue, we applied a shortest path algorithm to update the distance matrix and make it satisfy the triangle inequality. Then, we solved all SDVRPTW instances again using our algorithm. The optimal solution values of the 262 SDVRPTW instances obtained by previous articles can be found at: \\ \url{http://www.gerad.ca/~guyd/sdvrptw.html}.
Based on our computational results, we divided the instances into five categories:
\begin{description}
\item {\em Category} 1: The optimal solution values are smaller than those reported in \citet{Archetti2011};
\item {\em Category} 2: The optimal solutions are not reported in \citet{Archetti2011} but have been found by our algorithm;
\item {\em Category} 3: The optimal solutions were reported in \citet{Archetti2011} but have not been found by our algorithm;
\item {\em Category} 4: The optimal solution values are the same as those reported in \citet{Archetti2011};
\item {\em Category} 5: The optimal solutions have not been found by any algorithm.
\end{description}
Our algorithm achieved optimal solutions for 264 out of 504 SDVRPTW instances. Table \ref{tab:3} presents the detailed integer results for the 52 instances belonging to Categories 1 and 2, including the number of vehicles used (\# vehicles), the number of split customers (\# splits), the number of branch-and-bound nodes (\# nodes), the number of added cuts (\# cuts), the integer solution value (IP) and the consumed computation time (Time). The instances contained in Category 3 are R102-50-50, C103-50-100, C202-50-100, C205-50-100, C102-100-100, C205-100-100 and C206-100-100.

\begin{table*}[!htp]
\normalsize{\caption{Integer Solution Results for the SDVRPTW Instances in Categories 1 and 2.}
\label{tab:3}}
\begin{center}
\scalebox{0.5}{
\begin{tabular}{rccccccccc}
\hline
   &   Instance & $n$ &   $Q$ &  \# vehicles &    \# splits &    \# nodes &     \# cuts &         IP &       Time \\
\hline
\multirow{44}[0]{*}{Category 1}   &       C201 &         25 &         30 &         16 &          7 &         1  &        89  &     909.8  &      2.1  \\

                &       C202 &         25 &         30 &         16 &          7 &          1 &       136  &     909.8  &      2.8  \\

                &       C203 &         25 &         30 &         16 &          7 &         10 &       114  &     909.8  &      8.6  \\

                &       C204 &         25 &         30 &         16 &          7 &          8 &        63  &     909.8  &     10.2  \\

                &       C205 &         25 &         30 &         16 &          7 &          9 &       102  &     909.8  &     16.2  \\

                &       C206 &         25 &         30 &         16 &          6 &          1 &        52  &     909.8  &      2.3  \\

                &       C207 &         25 &         30 &         16 &          7 &          2 &        55  &     909.8  &      3.2  \\

                &       C201 &         25 &         50 &         10 &          3 &          9 &        40  &     601.2  &      4.8  \\

                &       C205 &         25 &         50 &         10 &          3 &         23 &        59  &     601.0  &      8.1  \\

                &       C206 &         25 &         50 &         10 &          2 &         27 &        49  &     601.0  &      9.4  \\

                &       C208 &         25 &         50 &         10 &          3 &         25 &        62  &     601.0  &     10.2  \\

                &      RC201 &         25 &         50 &         11 &          1 &        33  &       130  &     940.6  &     17.1  \\

                &       C101 &         25 &        100 &          5 &          0 &         1  &        12  &     291.8  &      1.6  \\

                &       C102 &         25 &        100 &          5 &          0 &         7  &        17  &     291.8  &      3.9  \\

              &       C105 &         25 &        100 &          5 &          0 &         1  &        11  &     291.8  &      3.4  \\

              &       C106 &         25 &        100 &          5 &          0 &         1  &        12  &     291.8  &      1.8  \\

              &       C107 &         25 &        100 &          5 &          0 &         1  &        13  &     291.8  &      1.8  \\

              &       C108 &         25 &        100 &          5 &          0 &         9  &        20  &     291.8  &      4.4  \\

              &       C201 &         25 &        100 &          5 &          1 &         1  &         5  &     363.5  &      1.6  \\

             &       C206 &         25 &        100 &          5 &          0 &         1  &        72  &     359.9  &      3.4  \\

            &       C207 &         25 &        100 &          5 &          1 &        17  &        37  &     358.7  &     13.8  \\

             &       C208 &         25 &        100 &          5 &          1 &         9  &        72  &     358.7  &     11.4  \\

              &      RC201 &         25 &        100 &          6 &          0 &         1  &         4  &     534.0  &      1.3  \\

             &      RC202 &         25 &        100 &          6 &          0 &         5  &        10  &     526.2  &      2.4  \\

              &       R101 &         50 &         50 &         15 &          3 &        61  &        63  &   1,190.7  &     31.9  \\

               &       C108 &         50 &         50 &         18 &          6 &         1  &       248  &   1,011.8  &     18.2  \\

               &       C201 &         50 &         50 &         18 &          8 &        15  &       200  &   1,159.4  &     99.4  \\

             &       C202 &         50 &         50 &         18 &          8 &         1  &       260  &   1,156.9  &     19.9  \\

           &       C203 &         50 &         50 &         18 &          8 &        15  &       273  &   1,156.9  &    268.2  \\

          &       C205 &         50 &         50 &         18 &          9 &         1  &       592  &   1,156.9  &     55.8  \\

         &       C206 &         50 &         50 &         18 &          8 &         5  &       558  &   1,156.9  &    313.1  \\

            &       C207 &         50 &         50 &         18 &          8 &        19  &       232  &   1,156.9  &    399.8  \\

             &       C208 &         50 &         50 &         18 &          9 &         5  &     1,076  &   1,156.9  &    483.2  \\

             &       R101 &         50 &        100 &         12 &          0 &         1  &         0  &   1,043.8  &      1.1  \\

             &       R105 &         50 &        100 &          9 &          0 &       316  &        39  &     918.1  &    173.2  \\

            &       R109 &         50 &        100 &          8 &          1 &     1,267  &       176  &     804.1  &  2,460.4  \\

               &       C101 &         50 &        100 &          9 &          0 &       111  &       451  &     587.5  &  1,487.4  \\

             &       C102 &         50 &        100 &          9 &          1 &        37  &       425  &     584.6  &  1,061.9  \\

              &       C105 &         50 &        100 &          9 &          2 &       115  &       385  &     587.5  &  1,901.9  \\

             &       C106 &         50 &        100 &          9 &          3 &       181  &       412  &     587.5  &  2,175.3  \\

             &       C107 &         50 &        100 &          9 &          3 &       201  &       139  &     587.5  &  1,710.7  \\

                &       C108 &         50 &        100 &          9 &          1 &         5  &       154  &     584.0  &     80.9  \\

              &       R205 &         50 &        100 &          8 &          1 &        49  &       120  &     758.8  &    418.1  \\

                &       R101 &        100 &        100 &         20 &          0 &        33  &         7  &   1,638.4  &     89.0  \\
\hline
\multirow{8}[0]{*}{Category 2}  &       C101 &         50 &         30 &         29 &         10 &         2  &       850  &   1,599.5  &     37.0  \\

                &       C102 &         50 &         30 &         29 &         10 &         1  &     2,689  &   1,599.5  &    260.3  \\

               &       C105 &         50 &         30 &         29 &          7 &         1  &     1,214  &   1,599.5  &     76.7  \\

               &       C106 &         50 &         30 &         29 &          7 &         1  &       683  &   1,599.5  &     21.1  \\

              &       C107 &         50 &         30 &         29 &          9 &         1  &       767  &   1,599.5  &     35.5  \\

            &       C108 &         50 &         30 &         29 &         10 &         2  &     1,547  &   1,598.3  &    137.7  \\

             &       C204 &         50 &         50 &         18 &          9 &         5  &       277  &   1,156.9  &    213.3  \\

               &       R201 &         50 &        100 &          8 &          0 &     1,557  &       231  &     843.0  &  2,709.8  \\
\hline
\end{tabular}
}
\end{center}
\end{table*}

\subsection{Results on the SCVRPTWL instances}
Next, we tried to solve all SCVRPTWL instances to optimality using our branch-and-price-and-cut algorithm. These instances use the updated distance matrix that satieties the triangle inequality. At the beginning of the algorithm, we solved the linear relaxation of the problem that contains all $n$ SMV inequalities and does not consider any $k$-path inequality. The linear relaxation results are reported in Table \ref{tab:scvrptwl-lp}, which show that a lower bound for each instance was achieved within the time limit.

\begin{table}[!htp]
\normalsize{\caption{Linear Relaxation Results on the SCVRPTWL Instances.}
\label{tab:scvrptwl-lp}}
\begin{center}
\scalebox{0.6}{
\begin{tabular}{ccccccccccc}
\hline
           &            &            & \multicolumn{ 2}{c}{Q = 30} &            & \multicolumn{ 2}{c}{Q = 50} &            & \multicolumn{ 2}{c}{Q = 100} \\
\cline{4-5} \cline{7-8} \cline{10-11}
         $n$ & Solomon group &    \# inst &  \# solved &       Time &            &  \# solved &       Time &            &  \# solved &       Time \\
\hline
        25 &         R1 &         12 &         12 &       0.6  &            &         12 &       0.8  &            &         12 &       0.8  \\

           &         C1 &          9 &          9 &       0.7  &            &          9 &       0.9  &            &          9 &       1.6  \\

           &        RC1 &          8 &          8 &       0.6  &            &          8 &       0.7  &            &          8 &       1.5  \\

        50 &         R1 &         12 &         12 &       2.0  &            &         12 &       3.0  &            &         12 &       3.9  \\

           &         C1 &          9 &          9 &       1.7  &            &          9 &       2.2  &            &          9 &       4.2  \\

           &        RC1 &          8 &          8 &       1.3  &            &          8 &       2.0  &            &          8 &       4.0  \\

       100 &         R1 &         12 &         12 &      14.9  &            &         12 &      23.3  &            &         12 &      65.0  \\

           &         C1 &          9 &          9 &       6.2  &            &          9 &       8.6  &            &          9 &      21.7  \\

           &        RC1 &          8 &          8 &       9.1  &            &          8 &      13.8  &            &          8 &      59.2  \\

        25 &         R2 &         11 &         11 &       0.8  &            &         11 &       1.1  &            &         11 &       1.4  \\

           &         C2 &          8 &          8 &       0.7  &            &          8 &       0.8  &            &          8 &       1.8  \\

           &        RC2 &          8 &          8 &       0.6  &            &          8 &       0.9  &            &          8 &       1.8  \\

        50 &         R2 &         11 &         11 &       2.6  &            &         11 &       4.7  &            &         11 &       8.9  \\

           &         C2 &          8 &          8 &       1.5  &            &          8 &       2.1  &            &          8 &       5.7  \\

           &        RC2 &          8 &          8 &       1.5  &            &          8 &       2.4  &            &          8 &       5.4  \\

       100 &         R2 &         11 &         11 &      21.3  &            &         11 &      41.1  &            &         11 &     168.9  \\

           &         C2 &          8 &          8 &       6.2  &            &          8 &      10.3  &            &          8 &      35.7  \\

           &        RC2 &          8 &          8 &      12.8  &            &          8 &     190.4  &            &          8 &     476.7  \\
\hline
\end{tabular}
}
\end{center}
\end{table}

Table \ref{tab:ip} presents a summary of the integer solution results of the SCVRPTWL instances. All columns except the first three columns give the average values over the optimally solved instances. We denote by LP and LPC the optimal values of the linear relaxations (at the root node) with and without the $k$-path inequalities, respectively. The value of ``LP gap (\%)'' (respectively, ``LPC gap (\%)'') for each solved instance was calculated by (IP $-$ LP)/IP (respectively, (IP $-$ LPC)/IP), where IP represents the optimal integer solution value. The average times used to produce LP, LPC and IP are reported in the columns ``LP time'', ``LPC time'' and ``IP time'', respectively.

\begin{table*}[!htp]
\normalsize{\caption{Summary of the Integer Solution Results to the SCVRPTWL Instances.}
\label{tab:ip}}
\begin{center}
\scalebox{0.6}{
\begin{tabular}{cccccccccccc}
\hline
Instance group &     \# inst &   \# solved & \# vehicles &   \# splits & LP gap (\%) &    LP time & LPC gap (\%) &   LPC time & IP time &    \# nodes &     \# cuts \\
\hline
  R1-25-30 &         12 &         12 &      13.4  &       1.8  &      0.39  &       0.6  &      0.13  &       1.4  &     117.2  &     356.5  &      64.2  \\

  R1-25-50 &         12 &         12 &       9.4  &       0.9  &      0.33  &       0.8  &      0.20  &       1.0  &       6.6  &      19.8  &       4.3  \\

 R1-25-100 &         12 &         12 &       7.3  &       0.4  &      0.10  &       0.8  &      0.10  &       0.8  &       1.1  &       1.7  &       0.0  \\

  C1-25-50 &          9 &          6 &      10.0  &       1.8  &      1.53  &       0.8  &      0.13  &       1.9  &      49.3  &     107.3  &      46.2  \\

 C1-25-100 &          9 &          9 &       5.0  &       0.0  &      0.38  &       1.6  &      0.08  &       2.5  &       5.2  &       3.2  &       8.1  \\

 RC1-25-50 &          8 &          8 &      12.0  &       1.1  &      2.39  &       0.7  &      0.24  &       1.8  &    1122.7  &    4144.9  &     218.4  \\

RC1-25-100 &          8 &          8 &       6.0  &       0.5  &      0.46  &       1.5  &      0.25  &       1.8  &       8.1  &      12.8  &       5.4  \\

  R2-25-30 &         11 &         11 &      13.2  &       2.2  &      0.54  &       0.8  &      0.18  &       1.7  &     284.3  &     308.8  &     107.9  \\

  R2-25-50 &         11 &         11 &       9.0  &       0.7  &      0.35  &       1.1  &      0.21  &       1.5  &      46.5  &      83.2  &       7.7  \\

 R2-25-100 &         11 &         11 &       6.4  &       0.7  &      0.04  &       1.4  &      0.04  &       1.4  &       5.3  &       4.8  &       0.0  \\

  C2-25-30 &          8 &          1 &      16.0  &       6.0  &      0.70  &       0.6  &      0.00  &       1.3  &       2.1  &       4.0  &      24.0  \\

  C2-25-50 &          8 &          8 &      10.0  &       1.1  &      0.24  &       0.8  &      0.00  &       1.4  &       1.4  &       1.0  &      20.6  \\

 C2-25-100 &          8 &          8 &       6.1  &       0.8  &      0.76  &       1.8  &      0.21  &       2.9  &     173.8  &      90.3  &      30.5  \\

 RC2-25-50 &          8 &          5 &      12.0  &       1.0  &      2.24  &       0.8  &      0.13  &       2.0  &     323.4  &     722.6  &      88.2  \\

RC2-25-100 &          8 &          8 &       6.0  &       1.0  &      0.50  &       1.8  &      0.27  &       2.2  &      38.3  &      15.9  &       6.3  \\

  R1-50-50 &         12 &          2 &      20.5  &       2.5  &      0.44  &       1.6  &      0.09  &       3.1  &     156.8  &      61.0  &      20.0  \\

 R1-50-100 &         12 &         10 &      12.7  &       0.5  &      0.10  &       4.0  &      0.08  &       5.2  &     219.0  &      26.4  &       4.2  \\

 C1-50-100 &          9 &          8 &      10.0  &       1.6  &      1.29  &       4.1  &      0.13  &       9.6  &     701.0  &      95.0  &      33.4  \\

 RC1-50-50 &          8 &          8 &      20.0  &       6.5  &      0.24  &       2.0  &      0.04  &       3.3  &     676.1  &     639.5  &      74.8  \\

RC1-50-100 &          8 &          5 &      10.2  &       1.6  &      0.53  &       4.3  &      0.11  &       7.5  &     479.3  &     171.4  &      13.8  \\

 R2-50-100 &         11 &          9 &      11.2  &       1.7  &      0.15  &       9.3  &      0.10  &      21.3  &     660.0  &      49.9  &      11.1  \\

 C2-50-100 &          8 &          1 &      10.0  &       4.0  &      0.95  &       5.6  &      0.03  &      11.1  &     153.3  &       5.0  &      60.0  \\

 RC2-50-50 &          8 &          8 &      20.0  &       6.1  &      0.24  &       2.4  &      0.03  &       4.1  &     640.0  &     413.4  &      55.5  \\

RC2-50-100 &          8 &          6 &      10.0  &       1.8  &      0.48  &       6.1  &      0.08  &      12.3  &      52.3  &       8.0  &       8.8  \\

R2-100-100 &         11 &          1 &      22.0  &       5.0  &      0.00  &      27.8  &      0.00  &      27.9  &      28.0  &       1.0  &       0.0  \\
\hline
\end{tabular}
}
\end{center}
\end{table*}

Our algorithm optimally solved 188 out of 504 SCVRPTWL instances, where 130, 57 and 1 instances contain 25, 50 and 100 customers, respectively. Further, we summarized the information on the number of solved instances and the number of split customers in Tables \ref{tab:solved} and \ref{tab:splits}. From Table \ref{tab:solved}, we can see that 48, 23, 29, 43, 18, 27 solved instances were derived from Solomon groups R1, C1, RC1, R2, C2 and RC2, respectively. Moreover, this table also clearly shows that the instances with greater $Q$ are easier to be solved. Table \ref{tab:splits} implies that with the increase of the vehicle capacity $Q$, the possibility of splitting customers becomes smaller. The detailed integer solution results for all optimally solved instances are given in Appendix C.
\begin{table}[!htp]
\normalsize{\caption{Summary on the Number of Solved Instances.}
\label{tab:solved}}
\begin{center}
\scalebox{0.6}{
\begin{tabular}{cccccc}
\hline
Solomon group &          $n$ &       $Q$ = 30 &       $Q$ = 50 &      $Q$ = 100 &        Sum \\
\hline
        R1 &         25 &         12 &         12 &         12 &         36 \\

           &         50 &          0 &          2 &         10 &         12 \\

           &        100 &          0 &          0 &          0 &          0 \\

        C1 &         25 &          0 &          6 &          9 &         15 \\

           &         50 &          0 &          0 &          8 &          8 \\

           &        100 &          0 &          0 &          0 &          0 \\

       RC1 &         25 &          0 &          8 &          8 &         16 \\

           &         50 &          0 &          8 &          5 &         13 \\

           &        100 &          0 &          0 &         0  &          0 \\

        R2 &         25 &         11 &         11 &         11 &         33 \\

           &         50 &          0 &          0 &          9 &          9 \\

           &        100 &          0 &          0 &          1 &          1 \\

        C2 &         25 &          1 &          8 &          8 &         17 \\

           &         50 &          0 &          0 &          1 &          1 \\

           &        100 &          0 &          0 &          0 &          0 \\

       RC2 &         25 &          0 &          5 &          8 &         13 \\

           &         50 &          0 &          8 &          6 &         14 \\

           &        100 &          0 &          0 &          0 &          0 \\
\hline
       Sum &            &        24  &        68  &        96  &            \\
\hline
\end{tabular}
}
\end{center}
\end{table}

\begin{table}[!htp]
\normalsize{\caption{Summary on the Number of Split Customers.}
\label{tab:splits}}
\begin{center}
\scalebox{0.6}{
\begin{tabular}{ccccc}
\hline
Solomon group &          $n$ &       $Q$ = 30 &       $Q$ = 50 &      $Q$ = 100 \\
\hline
        R1 &        25  &       1.8  &       0.9  &       0.4  \\

           &        50  &          -- &       2.5  &       0.5  \\

           &       100  &          -- &         --&       -- \\

        C1 &        25  &          -- &       1.8  &       0.0  \\

           &        50  &         -- &         -- &       1.6  \\

           &       100  &          -- &         -- &        -- \\

       RC1 &        25  &          -- &       1.1  &       0.5  \\

           &        50  &        -- &       6.5  &       1.6  \\

           &       100  &         -- &         -- &         -- \\

        R2 &        25  &       2.2  &       0.7  &       0.7  \\

           &        50  &          -- &          -- &       1.7  \\

           &       100  &          -- &          -- &       5.0  \\

        C2 &        25  &       6.0  &       1.1  &       0.8  \\

           &        50  &          -- &         -- &       4.0  \\

           &       100  &          --&         -- &         -- \\

       RC2 &        25  &          -- &       1.0  &       1.0  \\

           &        50  &         -- &       6.1  &       1.8  \\

           &       100  &          -- &          -- &         -- \\
\hline
\end{tabular}
}
\end{center}
\end{table}

\section{Conclusions}
\label{sec:c}
This paper introduces a new extension of the SDVRPTW in which the travel cost per unit distance is charged based on a linear function of the vehicle weight; this extension is called the split-collection vehicle routing problem with time windows and linear weight-related cost (SCVRPTWL). We devised an exact branch-and-price-and-cut algorithm to solve the problem, where the pricing subproblem is a resource-constrained elementary least-cost path problem. The effectiveness of the branch-and-price-and-cut algorithm heavily relies on the method for solving the pricing subproblem. We observed that at least one of the optimal solutions to the pricing subproblem must correspond to an extreme collection pattern; this help us reduce the feasible region significantly. To solve this new type of pricing subproblem, we designed a tailored and novel label-setting algorithm that integrates specific labels and dominance rules. We applied our branch-and-price-and-cut algorithm to solve the instances of both the SDVRPTW and SDVRPTWL.

The reported computational results reveal that our algorithm achieved optimal solutions for 264 SDVRPTW instances and 188 SDVRPTWL instances within one hour of computation time. The existing best exact algorithm, namely the enhanced branch-and-price-and-cut algorithm proposed by \citet{Archetti2011}, only produced optimal solutions for 262 SDVRPTW instances. Since the SDVRPTWL is a new problem and has not been tackled by any exiting algorithm, the experiments and analysis presented in this study serves as benchmarks for future researchers.

Since our branch-and-price-and-cut algorithm only optimally solved around one-third of the total benchmark SDVRPTWL instances, there is much space to improve the solution procedure. Furthermore, we may investigate other vehicle routing models that incorporates the linear weight-related cost or other types of cost functions, e.g., piecewise linear function of the vehicle weight.


\appendix

\section{Implementation Details of Set Dominance Rule}
We eliminated the dominated labels ending at vertex $i$ by maintaining a directed dominance graph $\mathbb{G}_i = (\mathbb{N}_i, \mathbb{A}_i)$. For the rest of this section, we distinguish between the terms  {\em vertex} and {\em node}, which are usually considered the same and are used interchangeably; we specify that {\em vertex} refers to the vertex in the underlying graph $G$ of the SCVRPTWL, and {\em node} refers to the node in the dominance graph $\mathbb{G}_i$.

Each node $u \in \mathbb{N}_i$ includes a set $L_{u}$ of non-dominated labels that end at vertex $i$ with the same $\tau_i$ and $V_i$ and has three attributes: the earliest service starting time $\tau_i(L_u)$, the set $V_i(L_u)$ of reachable vertices and the minimum reduced cost function $G_{min}(L_u, q) = \min_{E_i^x \in L_u}\{G^x(r^x, q)\}$ for $q \in [0, Q]$. The directed edge $(u, v)$ is included in the edge set $\mathbb{A}_i$ if there does not exist other paths from node $u$ to node $v$ and at least one of the following conditions holds:
\begin{enumerate}
 \item $\tau_i(L_u) \leq \tau_i(L_v)$ and $V_i(L_u) \supset V_i(L_v)$;
 \item $\tau_i(L_u) < \tau_i(L_v)$ and $V_i(L_u) \supseteq V_i(L_v)$.
 \end{enumerate}
We update $G_{min}(L_v, q) = \min\{G_{min}(L_u, q), G_{min}(L_v, q)\}$ with the creation of edge $(u, v)$. In this dominance graph, there must exist a root node 0 (a node that does not have incoming edges), which corresponds to the two-vertex partial route $r =(0, i)$. An example of the dominance graph is given in Figure \ref{fig:ex3}. Starting from the root node 0, we update $\mathbb{G}_i$ by invoking Algorithm \ref{alg:updateG} every time a new label ending at vertex $i$ is created. After performing this algorithm, the labels included in $\mathbb{G}_i$ are all currently non-dominated.
\begin{figure}[!th]
\begin{center}
\resizebox{8cm}{!}{\includegraphics{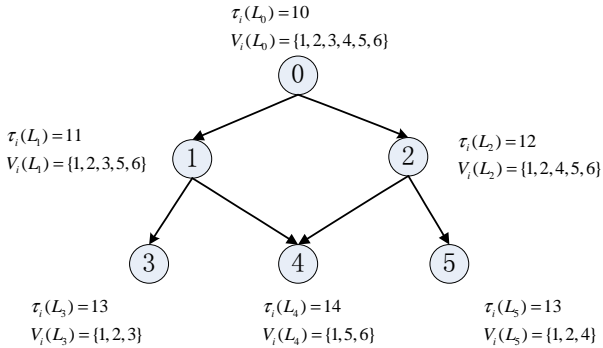}}
\end{center}
\caption{An example of the dominance graph $\mathbb{G}_i$.} \label{fig:ex3}
\end{figure}

\begin{algorithm}[!h]
\caption{The Process of Updating Graph $\mathbb{G}_i$.} \label{alg:updateG}
\begin{small}
\begin{algorithmic}[1]
\STATE INPUT: the current graph $\mathbb{G}_i$;
\STATE Check whether label $E^x_i$ is dominated; \label{alg:line:checkdom}
\IF{$E_i^x$ is not dominated}
\STATE Insert $E_i^x$ into the label set of a certain existing node or create a new node whose label set contains only $E_i^x$ \label{alg:insert};
\STATE Remove from $\mathbb{G}_i$ the previously inserted labels that become dominated after the adding of $E_i^x$ \label{alg:removeL};
\ELSE
\STATE Discard $E_i^x$ \label{alg:discard}.
\ENDIF
\end{algorithmic}
\end{small}
\end{algorithm}

We check whether a label $E^x_i = (\tau^x_i, N^x_i, V^x_i, G^x(r, q))$ is dominated (line \ref{alg:line:checkdom} in Algorithm \ref{alg:updateG}) by performing a recursive procedure {\em DominanceCheck}($u$, $E^x_i, \mathbb{G}_i$ ) shown in Algorithm \ref{alg:dominancecheck}. Given a node $u \in \mathbb{G}_i$, if there exists a child $v$ of $u$ that satisfies condition 1 or 2 (or both), the procedure moves to node $v$ since $G_{min}(L_v, q)$ has more chance to lie below $G^x(r^x, q)$. Otherwise, the procedure checks whether $E^x_i$ is a dominated label, i.e., whether $G_{min}(L_u, q) \leq G^x(r^x, q)$ holds for each $q \in [0, Q]$.
\begin{algorithm}[!h]
\caption{{\em DominanceCheck}($u, E^x_i, \mathbb{G}_i$).} \label{alg:dominancecheck}
\begin{small}
\begin{algorithmic}[1]
\STATE {\em flag} $\leftarrow$ false;
\FOR{each child $v$ of $u$ that has not been examined}
\IF{$\tau_i(L_v) \leq \tau^x_i, V_i(L_v) \supset V^x_i$ or $\tau_i(L_v) < \tau^x_i, V_i(L_v) \supseteq V^x_i$ }
\STATE {\em flag} $\leftarrow$ true;
\STATE {\em result} $\leftarrow$ {\em DominanceCheck}($v, E^x_i, \mathbb{G}_i$);
\IF{{\em result} = true}
\RETURN true;
\ENDIF
\ENDIF
\ENDFOR
\IF{{\em flag} = false}
\IF{$G_{min}(L_u, q) \leq G^x(r^x, q)$ for each $q \in [0, Q]$}
\RETURN true;
\ELSE
\RETURN false;
\ENDIF
\ENDIF
\RETURN false.
\end{algorithmic}
\end{small}
\end{algorithm}

If $E^x_i$ is a dominated label, we discard it (line \ref{alg:discard} in Algorithm \ref{alg:updateG}); otherwise, we need to add it in $\mathbb{G}_i$. We either insert the non-dominated $E^x_i$ into $L_u$ if $\tau_i(L_u) = \tau_i^x$ and $V_i(L_u) = V_i^x$ or create a new node whose label set contains only $E^x_i$ if no such $u$ exists (line \ref{alg:insert} in Algorithm \ref{alg:updateG}). The function {\em SearchNode}($u$, $E^x_i$, $\mathbb{G}_i$) presented in Algorithm \ref{alg:search} is used to check whether $\mathbb{G}_i$ contains a node $u$ with $\tau_i(L_u) = \tau_i^x$ and $V_i(L_u) = V_i^x$. After inserting $E^x_i$ into $L_u$, we update $G_{min}(L_u, q) = \min\{G_{min}(L_u, q), G^x(r^x, q)\}$ for all $q \in [0, Q]$. The new node is created and connected to $\mathbb{G}_i$ by invoking the function {\em CreateNode}($u$, $E_i^x$, $\mathbb{G}_i$ ) shown in Algorithm \ref{alg:create}. In this algorithm, we first create a node $w$ and then connect it to $\mathbb{G}_i$. If a node $u$ satisfies $\tau_i(L_u) \leq \tau_i(L_w)$ and $V_i(L_u) \supseteq V_i(L_w)$, and none of its child has this relationship, we create edge $(u, w)$. Then, for each child $v$ of $u$, if $\tau_i(L_v) \geq \tau_i(L_w)$ and $V_i(L_v) \subseteq V_i(L_w)$, we remove edge $(u, v)$ and create edge $(w, v)$. The newly created node $w$ may have multiple immediate predecessors and successors.

\begin{algorithm}[!h]
\caption{{\em SearchNode}($u$, $E^x_i$, $\mathbb{G}_i$).} \label{alg:search}
\begin{small}
\begin{algorithmic}[1]
\IF{$\tau_i(L_u) = \tau_i^x$ and $V_i(L_u) = V_i^x$}
\RETURN $u$;
\ENDIF
\FOR{Each child $v$ of $u$ that has not been examined}
\IF{$\tau_i(L_v) \leq \tau_i^x$ and $V_i(L_v) \supseteq V_i^x$}
\RETURN {\em SearchNode}($v$, $E^x_i$, $\mathbb{G}_i$);
\ENDIF
\ENDFOR
\RETURN null.
\end{algorithmic}
\end{small}
\end{algorithm}

\begin{algorithm}[!h]
\caption{{\em CreateNode}($u$, $E^x_i$, $\mathbb{G}_i$).} \label{alg:create}
\begin{small}
\begin{algorithmic}[1]
\STATE {\em flag} $\leftarrow$ false;
\STATE Create a node $w$ with $L_w=\{E^x_i\}$, $\tau_i(L_w) = \tau_i^x$, $V_i(L_w) = V_i^x$ and $G_{min}(L_w, q) = G^x(r^x, q)$;
\FOR{each child $v$ of $u$ that has not been examined}
\IF{$\tau_i(L_v) \leq \tau_i^x$ and $V_i(L_v) \supseteq V_i^x$ }
\STATE  {\em flag} $\leftarrow$ true;
\STATE {\em CreateNode}($v$, $E^x_i$, $\mathbb{G}_i$);
\ENDIF
\ENDFOR
\IF{{\em flag} = false}
\STATE Create edge $(u, w)$ and update $G_{min}(L_w, q) = \min\{G_{min}(L_u, q), G_{min}(L_w, q)\}$ for all $q \in [0, Q]$;
\FOR{each child $v$ of $u$}
\IF{$\tau_i(L_v) \geq \tau_i(L_w)$ and $V_i(L_v) \subseteq V_i(L_w)$}
\STATE Remove edge $(u, v)$ and create edge $(w, v)$.
\ENDIF
\ENDFOR
\ENDIF
\end{algorithmic}
\end{small}
\end{algorithm}

The adding of label $E^x_i$ in graph $\mathbb{G}_i$ necessitates the improvement of the minimal reduced cost function at all successor nodes. Moreover, some labels in $\mathbb{G}_i$ may become dominated due to the adjustment of the minimal reduced cost function and therefore can be removed (line \ref{alg:removeL} in Algorithm \ref{alg:updateG}). To improve the minimal reduced cost function and remove the dominated labels, we invoke a function {\em RemoveLabel}($u$, $\mathbb{G}_i$), which is shown in Algorithm \ref{alg:remove}.

\begin{algorithm}[!h]
\caption{{\em RemoveLabel}($u$, $\mathbb{G}_i$).} \label{alg:remove}
\begin{small}
\begin{algorithmic}[1]
\STATE {\em flag} $\leftarrow$ false;
\FOR{each child $v$ of $u$ that has not been examined}
\FOR{each label $E^x_i \in L_v$}
\IF{$G_{min}(L_u, q) \leq G^x(r^x, q)$ for all $q \in [0, Q]$}
\STATE Remove $E^x_i$ from $L_v$;
\ENDIF
\ENDFOR
\IF{$L_v = \emptyset$}
\STATE Create an edge from $u$ to each child of $v$;
\STATE Remove $v$ from $\mathbb{G}_i$;
\ELSE
\STATE $G_{min}(L_v, q) = \min{G_{min}(L_u, q), G_{min}(L_v, q)}$;
\ENDIF
\ENDFOR
\end{algorithmic}
\end{small}
\end{algorithm}

\section{The Detailed Derivation of $G^b(r, Q - \hat{q})$}
Given a feasible partial backward route $r = (v(1), v(2), \ldots, v(|r|))$, where $|r|\geq 2$ and $v(|r|) = n+1$, and the incoming flow $\hat{q}$, the reduced cost $G^b(r, Q - \hat{q})$ can be computed by the following model:
\begin{align}
G^b(r, Q - \hat{q}) =&\min \sum_{i = 1}^{|r|-1} c_{v(i), v(i+1)}\Big{(}a(\hat{q} + \sum_{j = 1}^i\delta_{v(j)}) + b \Big{)} \nonumber\\
&- \sum_{i = 1}^{|r|}\Big{(}\delta_{v(i)}\pi_{v(i)} + \mu_{v(i)}\Big{)}\\
\mbox{s.t.}~~ &\sum_{i = 1}^{|r|}\delta_{v(i)} = Q - \hat{q} \label{ahb:c1}\\
&0 \leq \delta_{v(i)} \leq d_{v(i)}, ~\forall~ 1 \leq  i \leq |r| \label{ahb:c2}
\end{align}

The objective can be rewritten as:
\begin{small}
\begin{align}
G^b(r, Q - \hat{q}) = & a\hat{q}\sum_{i = 1}^{|r| -1}c_{v(i), v(i+1)} + \sum_{i = 1}^{|r| -1}a\delta_{v(i)}\sum_{j = i}^{|r| -1}c_{v(j), v(j+1)}  \nonumber\\
&+ \sum_{i = 1}^{|r|-1}bc_{v(i), v(i+1)} - \sum_{i = 1}^{|r|}\bigg{(}\delta_{v(i)}\pi_{v(i)} + \mu_{v(i)}\bigg{)} \nonumber\\
= &a \bigg{(}Q - \sum_{i = 1}^{|r|}\delta_{v(i)}\bigg{)}\sum_{i = 1}^{|r| -1}c_{v(i), v(i+1)} + \sum_{i = 1}^{|r| -1}a\delta_{v(i)}\sum_{j = i}^{|r| -1}c_{v(j), v(j+1)}   \nonumber\\
&+ \sum_{i = 1}^{|r|-1}bc_{v(i), v(i+1)} - \sum_{i = 1}^{|r|}\bigg{(}\delta_{v(i)}\pi_{v(i)} + \mu_{v(i)}\bigg{)} \nonumber\\
=& aQ\sum_{i = 1}^{|r| - 1}c_{v(i), v(i+1)} - \sum_{i = 1}^{|r|} a\delta_{v(i)}\sum_{i = 1}^{|r| - 1}c_{v(i), v(i+1)} - \sum_{i = 1}^{|r|}\mu_{v(i)} \nonumber \\
& + \sum_{i = 1}^{|r| - 1}a\delta_{v(i)}\sum_{j = i}^{|r| -1}c_{v(j), v(j+1)}
+\sum_{i = 1}^{|r| - 1}bc_{v(i), v(i+1)} - \sum_{i = 1}^{|r|}\pi_{v(i)}\delta_{v(i)}\nonumber\\
= & aQ\sum_{i = 1}^{|r|-1}c_{v(i), v(i+1)} + \sum_{i = 1}^{|r|-1}bc_{v(i), v(i+1)}  \nonumber \\
&  - \sum_{i = 1}^{|r|}\mu_{v(i)} - \delta_{v(|r|)}a\sum_{i = 1}^{|r|-1}c_{v(i), v(i+1)} \nonumber\\
&- \sum_{i = 1}^{|r| -1}\delta_{v(i)}\bigg{(}a\sum_{i = 1}^{|r| -1}c_{v(i), v(i+1)} + \pi_{v(i)} - a\sum_{j =i}^{|r| -1}c_{v(j), v(j+1)}\bigg{)} \nonumber\\
=&\sum_{i = 1}^{|r| -1}\bigg{(}(aQ + b)c_{v(i), v(i+1)} - \mu_{v(i)}\bigg{)} - \delta_{v(|r|)}a\sum_{i = 1}^{|r|-1}c_{v(i), v(i+1)} \nonumber\\
&- \sum_{i = 1}^{|r| -1}\delta_{v(i)}\bigg{(}a\sum_{i = 1}^{|r| -1}c_{v(i), v(i+1)} + \pi_{v(i)} - a\sum_{j =i}^{|r| -1}c_{v(j), v(j+1)}\bigg{)}\nonumber\\
\end{align}
\end{small}

\section{The Detailed Integer Solution Results}
\begin{table*}[!htp]
\normalsize{\caption{Optimal Integer Solutions for the SCVRPTWL instances (Part I).}
\label{tab:ip1}}
\begin{center}
\scalebox{0.4}{
\begin{tabular}{cccccccccccc}
\hline
Instance group &   Instance & \# vehicles &   \# splits &         LP &    LP time &        LPC &   LPC time &         IP &    IP time &   \# nodes  &      \# cuts \\
\hline
  R1-25-30 &          1 &         15 &          1 &  15,623.9  &       0.4  &  15,682.2  &       0.7  &  15,682.2  &       0.7  &         1  &        20  \\

           &          2 &         15 &          3 &  15,043.8  &       0.7  &  15,089.9  &       1.2  &  15,223.3  &     660.7  &     3,141  &        98  \\

           &          3 &         13 &          2 &  14,753.6  &       0.6  &  14,781.8  &       1.3  &  14,786.3  &       3.8  &         9  &        98  \\

           &          4 &         13 &          2 &  14,724.9  &       0.7  &  14,760.3  &       1.8  &  14,769.6  &       4.1  &         7  &        72  \\

           &          5 &         14 &          0 &  15,334.8  &       0.4  &  15,383.8  &       0.6  &  15,383.8  &       0.6  &         1  &         8  \\

           &          6 &         13 &          1 &  14,778.5  &       0.5  &  14,804.1  &       1.0  &  14,805.9  &       1.6  &         3  &        26  \\

           &          7 &         13 &          2 &  14,745.3  &       0.8  &  14,781.5  &       1.5  &  14,786.3  &       3.9  &         9  &        64  \\

           &          8 &         13 &          2 &  14,721.0  &       0.8  &  14,760.3  &       2.0  &  14,769.6  &       4.9  &        11  &        53  \\

           &          9 &         13 &          2 &  14,806.8  &       0.5  &  14,830.7  &       1.0  &  14,832.5  &       1.5  &         3  &         8  \\

           &         10 &         13 &          2 &  14,555.4  &       0.6  &  14,594.6  &       1.6  &  14,599.8  &       3.6  &         9  &        64  \\

           &         11 &         13 &          2 &  14,773.2  &       0.7  &  14,801.2  &       1.5  &  14,805.8  &       3.1  &         7  &        38  \\

           &         12 &         13 &          3 &  14,481.0  &       0.9  &  14,547.6  &       2.2  &  14,599.7  &     718.2  &     1,077  &       221  \\
\hline
  R2-25-30 &          1 &         14 &          2 &  15,280.0  &       0.6  &  15,313.0  &       0.8  &  15,326.2  &       7.7  &        33  &        43  \\

           &          2 &         13 &          2 &  14,742.1  &       0.8  &  14,768.8  &       1.8  &  14,774.8  &       5.3  &        11  &        67  \\

           &          3 &         13 &          2 &  14,694.7  &       0.9  &  14,747.8  &       1.9  &  14,755.3  &       5.1  &        11  &        53  \\

           &          4 &         13 &          2 &  14,674.3  &       0.8  &  14,729.0  &       2.0  &  14,738.5  &       5.3  &        11  &        54  \\

           &          5 &         13 &          1 &  14,806.8  &       0.5  &  14,830.7  &       1.0  &  14,832.5  &       1.8  &         3  &         9  \\

           &          6 &         13 &          3 &  14,519.3  &       0.7  &  14,590.7  &       1.6  &  14,636.1  &     301.2  &       381  &       104  \\

           &          7 &         13 &          3 &  14,499.3  &       0.8  &  14,569.9  &       1.9  &  14,616.5  &     494.9  &       525  &       143  \\

           &          8 &         13 &          3 &  14,482.0  &       0.8  &  14,552.7  &       2.2  &  14,599.8  &     697.5  &       795  &       159  \\

           &          9 &         13 &          2 &  14,578.2  &       0.7  &  14,626.7  &       1.6  &  14,642.3  &       4.2  &         5  &       114  \\

           &         10 &         14 &          2 &  14,869.2  &       0.9  &  14,931.5  &       1.7  &  14,986.3  &      94.7  &       245  &        70  \\

           &         11 &         13 &          2 &  14,481.0  &       1.0  &  14,547.6  &       2.2  &  14,599.8  &   1,509.5  &     1,377  &       371  \\
\hline
  C2-25-30 &          1 &         16 &          6 &  19,549.8  &       0.6  &  19,693.0  &       1.3  &  19,693.0  &       2.1  &         4  &        24  \\
\hline
  R1-25-50 &          1 &         11 &          1 &  18,781.6  &       0.4  &  18,794.4  &       0.5  &  18,804.8  &       0.9  &         3  &         1  \\

           &          2 &         10 &          1 &  17,721.3  &       0.7  &  17,731.4  &       0.9  &  17,732.7  &       2.6  &         7  &         1  \\

           &          3 &          9 &          0 &  17,286.4  &       0.8  &  17,286.4  &       0.8  &  17,320.8  &       4.0  &        11  &         3  \\

           &          4 &          8 &          1 &  17,160.4  &       1.0  &  17,169.0  &       1.2  &  17,280.2  &      34.8  &       113  &        10  \\

           &          5 &         11 &          1 &  18,527.6  &       0.4  &  18,527.6  &       0.4  &  18,527.6  &       0.4  &         1  &         0  \\

           &          6 &          9 &          0 &  17,413.9  &       0.6  &  17,413.9  &       0.7  &  17,439.8  &       4.1  &        13  &         1  \\

           &          7 &          9 &          1 &  17,281.6  &       1.0  &  17,281.6  &       1.0  &  17,305.6  &       2.7  &         5  &         0  \\

           &          8 &          9 &          2 &  17,160.4  &       1.0  &  17,169.0  &       1.3  &  17,243.9  &       8.6  &        23  &        11  \\

           &          9 &          9 &          1 &  17,632.7  &       0.7  &  17,736.7  &       1.2  &  17,754.0  &       8.0  &        31  &         5  \\

           &         10 &         10 &          1 &  17,144.0  &       0.8  &  17,234.2  &       1.4  &  17,291.5  &       9.0  &        27  &        14  \\

           &         11 &          9 &          1 &  17,320.8  &       0.8  &  17,320.8  &       0.8  &  17,320.8  &       0.8  &         1  &         0  \\

           &         12 &          9 &          1 &  16,737.3  &       1.1  &  16,778.3  &       1.9  &  16,832.3  &       2.9  &         3  &         5  \\
\hline
  C1-25-50 &          1 &         10 &          0 &  18,501.2  &       0.8  &  18,819.4  &       1.7  &  18,841.8  &       4.3  &         9  &        35  \\

           &          2 &         10 &          1 &  18,333.4  &       0.9  &  18,549.7  &       2.2  &  18,591.8  &     239.8  &       537  &        77  \\

           &          5 &         10 &          3 &  18,501.2  &       0.7  &  18,819.4  &       1.5  &  18,841.8  &      13.1  &        31  &        40  \\

           &          6 &         10 &          4 &  18,501.2  &       0.7  &  18,819.4  &       1.6  &  18,841.8  &       4.5  &         9  &        34  \\

           &          7 &         10 &          1 &  18,393.6  &       0.6  &  18,694.1  &       1.7  &  18,707.8  &      10.5  &        19  &        43  \\

           &          9 &         10 &          2 &  18,119.2  &       1.2  &  18,231.9  &       2.4  &  18,250.2  &      23.4  &        39  &        48  \\
\hline
 RC1-25-50 &          1 &         12 &          1 &  32,263.0  &       0.6  &  32,757.9  &       1.4  &  32,874.0  &     261.1  &     1,523  &       102  \\

           &          2 &         12 &          0 &  31,791.8  &       0.7  &  32,502.3  &       1.4  &  32,508.5  &       2.0  &         3  &        25  \\

           &          3 &         12 &          2 &  31,589.4  &       0.7  &  32,348.2  &       2.0  &  32,408.5  &     410.9  &     1,463  &        73  \\

           &          4 &         12 &          2 &  31,571.1  &       0.8  &  32,321.5  &       2.0  &  32,405.8  &   1,777.1  &     5,547  &       236  \\

           &          5 &         12 &          1 &  31,967.2  &       0.7  &  32,611.2  &       1.6  &  32,778.0  &   3,224.3  &    13,322  &       204  \\

           &          6 &         12 &          1 &  31,821.6  &       0.7  &  32,544.3  &       1.5  &  32,573.3  &      21.8  &        89  &        52  \\

           &          7 &         12 &          1 &  31,435.1  &       0.7  &  32,197.1  &       2.1  &  32,268.8  &     603.6  &     2,449  &       321  \\

           &          8 &         12 &          1 &  31,435.1  &       0.9  &  32,186.6  &       2.4  &  32,268.8  &   2,680.4  &     8,763  &       734  \\
\hline
  R2-25-50 &          1 &         11 &          0 &  18,327.3  &       0.7  &  18,327.3  &       0.7  &  18,327.3  &       0.7  &         1  &         0  \\

           &          2 &          9 &          0 &  17,144.3  &       1.0  &  17,144.3  &       1.0  &  17,144.3  &       1.0  &         1  &         0  \\

           &          3 &          8 &          1 &  17,024.3  &       1.2  &  17,026.5  &       1.6  &  17,056.9  &      14.2  &        25  &         7  \\

           &          4 &          8 &          1 &  16,895.1  &       1.3  &  16,896.4  &       1.6  &  16,927.4  &      54.9  &        73  &         3  \\

           &          5 &          9 &          1 &  17,523.7  &       1.1  &  17,617.9  &       1.6  &  17,681.3  &      12.7  &        33  &         5  \\

           &          6 &          9 &          1 &  16,920.5  &       1.2  &  16,920.5  &       1.2  &  16,920.5  &       1.2  &         1  &         0  \\

           &          7 &          8 &          2 &  16,848.6  &       0.9  &  16,863.3  &       2.0  &  16,863.3  &       2.0  &         1  &         8  \\

           &          8 &          9 &          1 &  16,737.3  &       1.6  &  16,778.3  &       2.2  &  16,832.3  &       3.2  &         3  &        12  \\

           &          9 &          9 &          0 &  17,100.8  &       1.1  &  17,176.7  &       1.7  &  17,333.9  &     416.0  &       771  &        39  \\

           &         10 &         10 &          1 &  17,484.8  &       1.2  &  17,486.2  &       1.4  &  17,486.2  &       1.4  &         1  &         1  \\

           &         11 &          9 &          0 &  16,737.3  &       1.1  &  16,778.3  &       1.7  &  16,832.3  &       3.9  &         5  &        10  \\
\hline
\end{tabular}
}
\end{center}
\end{table*}

\begin{table*}[!htp]
\normalsize{\caption{Optimal Integer Solutions for the SCVRPTWL instances (Part II).}
\label{tab:ip2}}
\begin{center}
\scalebox{0.4}{
\begin{tabular}{cccccccccccc}
\hline
Instance group &   Instance & \# vehicles &   \# splits &         LP &    LP time &        LPC &   LPC time &         IP &    IP time &   \# nodes  &      \# cuts \\
\hline
  C2-25-50 &          1 &         10 &          2 &  20,851.5  &       0.6  &  20,914.0  &       0.8  &  20,914.0  &       0.8  &         1  &         4  \\

           &          2 &         10 &          2 &  20,700.3  &       0.8  &  20,762.8  &       1.1  &  20,762.8  &       1.2  &         1  &         4  \\

           &          3 &         10 &          1 &  20,387.4  &       0.9  &  20,424.3  &       1.3  &  20,424.3  &       1.3  &         1  &        13  \\

           &          4 &         10 &          1 &  20,359.8  &       1.0  &  20,424.3  &       2.0  &  20,424.3  &       2.0  &         1  &        42  \\

           &          5 &         10 &          0 &  20,563.0  &       0.8  &  20,602.5  &       1.1  &  20,602.5  &       1.1  &         1  &        10  \\

           &          6 &         10 &          2 &  20,531.2  &       0.7  &  20,569.5  &       1.4  &  20,569.5  &       1.4  &         1  &        21  \\

           &          7 &         10 &          1 &  20,529.5  &       0.9  &  20,575.5  &       1.6  &  20,575.5  &       1.6  &         1  &        34  \\

           &          8 &         10 &          0 &  20,378.2  &       0.8  &  20,424.3  &       1.5  &  20,424.3  &       1.5  &         1  &        37  \\
\hline
 RC2-25-50 &          1 &         12 &          1 &  32,224.5  &       0.8  &  32,745.9  &       1.6  &  32,820.3  &     179.8  &       567  &        96  \\

           &          2 &         12 &          0 &  31,751.7  &       0.8  &  32,502.3  &       1.7  &  32,508.5  &       4.6  &         9  &        39  \\

           &          3 &         12 &          2 &  31,589.4  &       0.9  &  32,354.8  &       2.3  &  32,408.5  &   1,108.5  &     2,099  &       104  \\

           &          6 &         12 &          1 &  31,883.9  &       1.0  &  32,544.3  &       2.1  &  32,573.3  &      41.7  &        99  &        54  \\

           &          7 &         12 &          1 &  31,490.6  &       0.7  &  32,215.9  &       2.2  &  32,268.8  &     282.5  &       839  &       148  \\
\hline
 R1-25-100 &          1 &         10 &          0 &  27,186.2  &       0.4  &  27,186.2  &       0.4  &  27,186.2  &       0.4  &         1  &         0  \\

           &          2 &          9 &          0 &  25,055.3  &       0.5  &  25,055.3  &       0.5  &  25,187.5  &       1.4  &         3  &         0  \\

           &          3 &          7 &          0 &  23,995.9  &       0.8  &  23,995.9  &       0.8  &  23,995.9  &       0.9  &         1  &         0  \\

           &          4 &          6 &          2 &  23,431.6  &       1.2  &  23,431.6  &       1.2  &  23,431.6  &       1.2  &         1  &         0  \\

           &          5 &          9 &          0 &  26,478.4  &       0.4  &  26,478.4  &       0.4  &  26,532.2  &       1.1  &         3  &         0  \\

           &          6 &          7 &          0 &  24,226.1  &       0.5  &  24,226.1  &       0.6  &  24,226.1  &       0.6  &         1  &         0  \\

           &          7 &          6 &          2 &  23,848.3  &       1.0  &  23,848.3  &       1.0  &  23,848.3  &       1.0  &         1  &         0  \\

           &          8 &          6 &          1 &  23,374.8  &       1.4  &  23,374.8  &       1.4  &  23,374.8  &       1.4  &         1  &         0  \\

           &          9 &          8 &          0 &  24,289.9  &       0.6  &  24,289.9  &       0.6  &  24,289.9  &       0.6  &         1  &         0  \\

           &         10 &          7 &          0 &  23,586.3  &       0.8  &  23,586.3  &       0.8  &  23,684.1  &       2.1  &         5  &         0  \\

           &         11 &          6 &          0 &  23,758.5  &       0.9  &  23,758.5  &       0.9  &  23,758.5  &       0.9  &         1  &         0  \\

           &         12 &          6 &          0 &  23,019.5  &       1.3  &  23,019.5  &       1.3  &  23,019.5  &       1.3  &         1  &         0  \\
\hline
 C1-25-100 &          1 &          5 &          0 &  21,018.1  &       1.0  &  21,036.0  &       1.1  &  21,036.0  &       1.1  &         1  &         4  \\

           &          2 &          5 &          0 &  20,704.3  &       1.6  &  20,713.5  &       1.8  &  20,713.5  &       1.8  &         1  &         3  \\

           &          3 &          5 &          0 &  20,525.3  &       2.1  &  20,645.8  &       3.6  &  20,713.5  &       9.2  &         5  &        16  \\

           &          4 &          5 &          0 &  20,231.0  &       2.9  &  20,334.5  &       7.0  &  20,388.5  &      21.0  &        13  &        17  \\

           &          5 &          5 &          0 &  21,018.1  &       0.9  &  21,036.0  &       1.1  &  21,036.0  &       1.1  &         1  &         4  \\

           &          6 &          5 &          0 &  21,018.1  &       1.1  &  21,036.0  &       1.5  &  21,036.0  &       1.6  &         1  &         4  \\

           &          7 &          5 &          0 &  21,017.2  &       0.9  &  21,036.0  &       1.3  &  21,036.0  &       1.4  &         1  &         6  \\

           &          8 &          5 &          0 &  20,736.3  &       1.5  &  20,944.7  &       2.3  &  20,965.5  &       7.5  &         5  &        13  \\

           &          9 &          5 &          0 &  20,430.0  &       2.0  &  20,490.0  &       2.5  &  20,490.0  &       2.5  &         1  &         6  \\
\hline
RC1-25-100 &          1 &          6 &          0 &  37,073.7  &       0.7  &  37,079.3  &       0.8  &  37,360.5  &       3.0  &         7  &         5  \\

           &          2 &          6 &          0 &  35,816.1  &       1.6  &  36,002.3  &       2.2  &  36,147.5  &      41.9  &        75  &        12  \\

           &          3 &          6 &          0 &  34,626.7  &       1.8  &  34,650.8  &       1.9  &  34,673.5  &       2.8  &         3  &         4  \\

           &          4 &          6 &          2 &  34,510.0  &       1.9  &  34,650.8  &       2.3  &  34,673.5  &       3.5  &         3  &         7  \\

           &          5 &          6 &          1 &  36,360.2  &       1.1  &  36,428.9  &       1.3  &  36,652.0  &       5.2  &         7  &         4  \\

           &          6 &          6 &          0 &  35,808.6  &       1.1  &  35,818.0  &       1.5  &  35,818.0  &       1.5  &         1  &         2  \\

           &          7 &          6 &          0 &  34,349.9  &       1.6  &  34,431.0  &       1.9  &  34,431.0  &       1.9  &         1  &         4  \\

           &          8 &          6 &          1 &  34,052.1  &       2.2  &  34,139.3  &       2.6  &  34,162.0  &       5.3  &         5  &         5  \\
\hline
 R2-25-100 &          1 &          8 &          0 &  25,812.8  &       0.9  &  25,812.8  &       0.9  &  25,812.8  &       0.9  &         1  &         0  \\

           &          2 &          6 &          0 &  23,653.9  &       1.0  &  23,653.9  &       1.0  &  23,663.1  &       3.6  &         5  &         0  \\

           &          3 &          6 &          0 &  23,394.3  &       1.4  &  23,394.3  &       1.4  &  23,486.2  &      39.7  &        37  &         0  \\

           &          4 &          6 &          1 &  23,166.9  &       1.7  &  23,166.9  &       1.7  &  23,166.9  &       1.7  &         1  &         0  \\

           &          5 &          7 &          1 &  23,921.2  &       1.1  &  23,921.2  &       1.1  &  23,921.2  &       1.1  &         1  &         0  \\

           &          6 &          6 &          2 &  23,014.5  &       1.4  &  23,014.5  &       1.4  &  23,014.5  &       1.4  &         1  &         0  \\

           &          7 &          6 &          1 &  22,896.7  &       1.4  &  22,896.7  &       1.4  &  22,896.7  &       1.4  &         1  &         0  \\

           &          8 &          6 &          2 &  22,630.8  &       1.6  &  22,630.8  &       1.6  &  22,630.8  &       1.6  &         1  &         0  \\

           &          9 &          7 &          0 &  23,462.1  &       1.5  &  23,462.1  &       1.5  &  23,462.1  &       1.6  &         1  &         0  \\

           &         10 &          6 &          0 &  23,901.0  &       1.4  &  23,901.0  &       1.4  &  23,907.3  &       3.3  &         3  &         0  \\

           &         11 &          6 &          1 &  22,630.8  &       2.0  &  22,630.8  &       2.0  &  22,630.8  &       2.0  &         1  &         0  \\
\hline
 C2-25-100 &          1 &          7 &          0 &  24,246.7  &       1.3  &  24,325.5  &       2.4  &  24,325.5  &       2.4  &         1  &        29  \\

           &          2 &          6 &          1 &  23,430.3  &       1.8  &  23,485.3  &       2.9  &  23,509.0  &      11.9  &        11  &         4  \\

           &          3 &          6 &          1 &  23,175.7  &       1.8  &  23,217.5  &       2.7  &  23,217.5  &       2.7  &         1  &         4  \\

           &          4 &          5 &          0 &  22,789.1  &       2.4  &  22,898.4  &       2.9  &  23,087.0  &     853.6  &       303  &        31  \\

           &          5 &          7 &          0 &  24,064.8  &       1.5  &  24,156.0  &       3.1  &  24,156.0  &       3.1  &         1  &        44  \\

           &          6 &          6 &          0 &  23,763.3  &       1.8  &  23,977.3  &       3.1  &  24,041.5  &      11.5  &         9  &        36  \\

           &          7 &          6 &          2 &  23,564.2  &       1.7  &  23,787.6  &       2.9  &  23,854.0  &     483.1  &       375  &        58  \\

           &          8 &          6 &          2 &  23,478.3  &       1.7  &  23,702.8  &       3.3  &  23,756.5  &      21.8  &        21  &        38  \\
\hline
\end{tabular}
}
\end{center}
\end{table*}

\begin{table*}[!htp]
\normalsize{\caption{Optimal Integer Solutions for the SCVRPTWL instances (Part III).}
\label{tab:ip3}}
\begin{center}
\scalebox{0.4}{
\begin{tabular}{cccccccccccc}
\hline
Instance group &   Instance & \# vehicles &   \# splits &         LP &    LP time &        LPC &   LPC time &         IP &    IP time &   \# nodes  &      \# cuts \\
\hline
RC2-25-100 &          1 &          6 &          1 &  37,193.7  &       1.4  &  37,239.5  &       1.8  &  37,509.0  &       3.4  &         3  &        13  \\

           &          2 &          6 &          3 &  35,666.7  &       1.7  &  35,967.9  &       2.6  &  36,199.5  &     280.5  &       105  &        20  \\

           &          3 &          6 &          0 &  34,499.0  &       1.9  &  34,502.3  &       2.2  &  34,525.0  &       3.3  &         3  &         1  \\

           &          4 &          6 &          1 &  34,425.9  &       2.0  &  34,502.3  &       2.7  &  34,525.0  &       4.3  &         3  &         7  \\

           &          5 &          6 &          2 &  35,916.0  &       1.9  &  35,965.8  &       2.5  &  36,160.0  &       5.1  &         4  &         6  \\

           &          6 &          6 &          0 &  35,980.0  &       1.5  &  35,980.0  &       1.5  &  35,980.0  &       1.5  &         1  &         0  \\

           &          7 &          6 &          0 &  34,302.9  &       1.5  &  34,395.8  &       1.7  &  34,407.0  &       4.2  &         5  &         1  \\

           &          8 &          6 &          1 &  34,052.1  &       2.3  &  34,139.3  &       2.9  &  34,162.0  &       4.2  &         3  &         2  \\
\hline
  R1-50-50 &          1 &         23 &          0 &  38,938.6  &       0.9  &  39,097.1  &       1.4  &  39,109.5  &       3.7  &         3  &         6  \\

           &          9 &         18 &          5 &  36,422.7  &       2.2  &  36,525.0  &       4.8  &  36,582.8  &     309.9  &       119  &        34  \\
\hline
 RC1-50-50 &          1 &         20 &          5 &  61,701.4  &       1.6  &  61,934.9  &       2.5  &  62,019.3  &     171.5  &       233  &        51  \\

           &          2 &         20 &          5 &  61,251.1  &       1.9  &  61,384.9  &       3.2  &  61,385.5  &      12.3  &        11  &        37  \\

           &          3 &         20 &          8 &  61,038.8  &       2.1  &  61,190.6  &       3.3  &  61,206.8  &      24.5  &        19  &        21  \\

           &          4 &         20 &          9 &  60,852.7  &       2.3  &  60,908.8  &       3.6  &  60,928.8  &   1,252.9  &     1,204  &       159  \\

           &          5 &         20 &          3 &  61,443.5  &       2.1  &  61,620.0  &       3.2  &  61,620.0  &       4.1  &         2  &        20  \\

           &          6 &         20 &          8 &  61,389.3  &       1.7  &  61,489.3  &       2.7  &  61,536.3  &   1,788.4  &     1,559  &       246  \\

           &          7 &         20 &          6 &  60,784.2  &       2.3  &  60,861.2  &       3.7  &  60,880.8  &     971.1  &       993  &        33  \\

           &          8 &         20 &          8 &  60,708.8  &       2.3  &  60,772.2  &       4.2  &  60,791.8  &   1,183.7  &     1,095  &        31  \\
\hline
 RC2-50-50 &          1 &         20 &          3 &  61,697.5  &       1.6  &  61,922.8  &       2.7  &  61,965.5  &     114.2  &        85  &        38  \\

           &          2 &         20 &          5 &  61,251.1  &       2.1  &  61,382.6  &       4.7  &  61,385.5  &      30.4  &        15  &        55  \\

           &          3 &         20 &          5 &  61,038.8  &       2.6  &  61,186.5  &       4.9  &  61,206.8  &      39.7  &        21  &        31  \\

           &          4 &         20 &         11 &  60,852.7  &       2.6  &  60,906.8  &       4.2  &  60,928.8  &   2,884.7  &     1,867  &       166  \\

           &          5 &         20 &          3 &  61,481.0  &       2.6  &  61,694.3  &       4.4  &  61,700.0  &      74.4  &        39  &        71  \\

           &          6 &         20 &          7 &  61,482.9  &       2.7  &  61,580.2  &       3.5  &  61,598.8  &     124.7  &       101  &        21  \\

           &          7 &         20 &          6 &  60,996.3  &       2.0  &  61,104.3  &       3.8  &  61,122.5  &     224.7  &       147  &        29  \\

           &          8 &         20 &          9 &  60,708.8  &       2.7  &  60,772.2  &       4.9  &  60,791.8  &   1,627.6  &     1,032  &        33  \\
\hline
 R1-50-100 &          1 &         18 &          0 &  53,915.2  &       1.0  &  53,915.2  &       1.1  &  53,933.9  &       5.3  &         5  &         0  \\

           &          2 &         14 &          0 &  49,497.6  &       2.2  &  49,497.6  &       2.3  &  49,498.7  &      24.3  &         9  &         0  \\

           &          3 &         12 &          0 &  46,047.4  &       2.7  &  46,047.4  &       2.7  &  46,047.4  &       2.8  &         1  &         0  \\

           &          4 &         10 &          2 &  43,178.8  &       6.3  &  43,183.7  &      10.6  &  43,271.5  &   1,209.6  &       103  &         7  \\

           &          5 &         16 &          1 &  51,958.0  &       1.1  &  51,987.1  &       1.4  &  51,987.1  &       1.4  &         1  &         1  \\

           &          6 &         12 &          1 &  47,578.2  &       2.8  &  47,578.2  &       2.8  &  47,625.0  &     283.5  &        85  &         1  \\

           &          7 &         12 &          0 &  44,839.0  &       2.7  &  44,839.0  &       2.8  &  44,980.5  &      75.6  &        13  &         1  \\

           &          8 &         10 &          1 &  43,059.7  &       6.2  &  43,073.1  &      10.1  &  43,138.6  &     372.3  &        35  &        26  \\

           &          9 &         13 &          0 &  48,334.6  &       3.0  &  48,334.6  &       3.0  &  48,334.6  &       3.0  &         1  &         0  \\

           &         12 &         10 &          0 &  42,412.5  &      11.4  &  42,445.2  &      15.5  &  42,458.6  &     212.2  &        11  &         6  \\
\hline
 C1-50-100 &          1 &         10 &          2 &  41,214.0  &       2.8  &  41,899.5  &       5.6  &  41,928.0  &     103.6  &        21  &        23  \\

           &          2 &         10 &          1 &  40,896.7  &       3.1  &  41,144.0  &       5.3  &  41,144.0  &       5.4  &         1  &         9  \\

           &          4 &         10 &          3 &  39,487.9  &       9.8  &  39,757.5  &      29.2  &  39,843.0  &     341.2  &        11  &        78  \\

           &          5 &         10 &          1 &  41,210.2  &       2.7  &  41,844.4  &       4.9  &  41,927.5  &   1,489.7  &       233  &        15  \\

           &          6 &         10 &          2 &  41,210.8  &       3.2  &  41,898.0  &       6.9  &  41,928.0  &     460.7  &        77  &        22  \\

           &          7 &         10 &          1 &  41,170.9  &       2.4  &  41,799.9  &       6.2  &  41,883.0  &   1,693.3  &       243  &        17  \\

           &          8 &         10 &          1 &  40,844.2  &       3.6  &  41,388.5  &       7.6  &  41,453.5  &   1,439.8  &       167  &        87  \\

           &          9 &         10 &          2 &  39,952.5  &       4.9  &  40,120.8  &      11.3  &  40,159.0  &      74.1  &         7  &        16  \\
\hline
RC1-50-100 &          4 &         10 &          1 &  64,532.7  &       5.9  &  64,836.0  &      11.3  &  64,836.0  &      11.3  &         1  &         5  \\

           &          5 &         11 &          0 &  68,401.8  &       2.7  &  68,737.5  &       6.0  &  69,051.5  &   2,357.4  &       851  &        52  \\

           &          6 &         10 &          0 &  67,988.5  &       3.6  &  68,245.8  &       4.8  &  68,297.0  &      12.4  &         3  &         5  \\

           &          7 &         10 &          5 &  65,191.2  &       4.2  &  65,478.5  &       6.5  &  65,478.5  &       6.6  &         1  &         4  \\

           &          8 &         10 &          2 &  63,864.5  &       5.1  &  64,080.5  &       8.8  &  64,080.5  &       8.8  &         1  &         3  \\
\hline
 R2-50-100 &          1 &         14 &          2 &  50,641.3  &       2.0  &  50,668.8  &       2.6  &  50,710.8  &     204.8  &        53  &         4  \\

           &          2 &         12 &          5 &  47,217.2  &       6.7  &  47,217.2  &       6.7  &  47,268.6  &     622.7  &        97  &         0  \\

           &          4 &         10 &          1 &  42,328.4  &      15.8  &  42,362.4  &      54.9  &  42,426.9  &     291.6  &        17  &        41  \\

           &          5 &         12 &          1 &  47,252.7  &       3.2  &  47,252.7  &       3.2  &  47,279.3  &      49.9  &         7  &         0  \\

           &          6 &         11 &          0 &  45,010.2  &       4.2  &  45,010.2  &       4.2  &  45,013.9  &      50.4  &         3  &         0  \\

           &          7 &         10 &          2 &  43,426.0  &       8.6  &  43,459.2  &      14.6  &  43,569.4  &   2,578.0  &       191  &         4  \\

           &          8 &         10 &          2 &  42,327.9  &      15.7  &  42,362.4  &      33.9  &  42,426.9  &   1,637.5  &        53  &        35  \\

           &         10 &         12 &          2 &  45,272.1  &       9.7  &  45,293.7  &      20.9  &  45,319.1  &     194.3  &        15  &         4  \\

           &         11 &         10 &          0 &  42,390.3  &      17.7  &  42,435.0  &      50.8  &  42,458.6  &     310.9  &        13  &        12  \\
\hline
 C2-50-100 &          6 &         10 &          4 &  45,450.0  &       5.6  &  45,876.3  &      11.1  &  45,888.0  &     153.3  &         5  &        60  \\
\hline
RC2-50-100 &          3 &         10 &          1 &  64,949.3  &       7.7  &  65,317.5  &      17.3  &  65,317.5  &      17.4  &         1  &         6  \\

           &          4 &         10 &          4 &  64,330.8  &       6.1  &  64,561.5  &       9.7  &  64,561.5  &       9.7  &         1  &         3  \\

           &          5 &         10 &          1 &  67,482.7  &       4.4  &  67,822.1  &      11.4  &  68,135.0  &     250.8  &        43  &        23  \\

           &          6 &         10 &          0 &  67,780.8  &       4.8  &  67,971.0  &       9.8  &  67,971.0  &       9.9  &         1  &        10  \\

           &          7 &         10 &          1 &  65,439.3  &       6.4  &  65,682.5  &      11.0  &  65,682.5  &      11.1  &         1  &         8  \\

           &          8 &         10 &          4 &  63,854.6  &       7.0  &  64,080.5  &      14.6  &  64,080.5  &      14.7  &         1  &         3  \\
\hline
R2-100-100 &          1 &         22 &          5 &  84,051.5  &      27.8  &  84,051.5  &      27.9  &  84,051.5  &      28.0  &         1  &         0  \\
\hline
\end{tabular}
}
\end{center}
\end{table*}


\end{document}